\documentclass[12pt]{amsart}
\usepackage[margin=2cm]{geometry}
\usepackage{amsfonts,mathrsfs}
\usepackage{amssymb}
\usepackage[dvips]{graphics}
\usepackage{epsfig}
\usepackage{euscript}
\usepackage{color}
\usepackage{soul}
\usepackage[colorlinks,pdfstartview=FitB]{hyperref}
\hypersetup{allcolors=blue}

\newtheorem{thm}{Theorem}[section]

\newtheorem{lem}[thm]{Lemma}

\newtheorem{ex}[thm]{Example}

\theoremstyle{definition}

\theoremstyle{remark}
\newtheorem{remark}[thm]{Remark}
\numberwithin{thm}{section}

\DeclareMathOperator{\IM}{Im}

\newcommand{\R}{{\mathord{\mathbb R}}}

\newcommand{\F}{\mathscr{F}}

\newcommand{\C}{{\mathord{\mathbb C}}}
\newcommand{\Z}{{\mathord{\mathbb Z}}}

\newcommand{\e}{\mathfrak{e}}

\def\idty{{\mathchoice {\mathrm{1\mskip-4mu l}} {\mathrm{1\mskip-4mu l}} %
{\mathrm{1\mskip-4.5mu l}} {\mathrm{1\mskip-5mu l}}}}

\DeclareMathOperator{\Tr}{Tr}
\DeclareMathOperator{\Span}{span}

\begin{document}

\title[]{Slow propagation velocities in Schr\"odinger operators with large periodic potential}

\author[H. Abdul-Rahman]{Houssam Abdul-Rahman}
\address{[H. Abdul-Rahman] Department of Mathematical Sciences, College of Science, United Arab Emirates University, 15551
Al Ain, UAE}
\email{\href{mailto:houssam.a@uaeu.ac.ae}{houssam.a@uaeu.ac.ae}}

\author[M.\ Darras]{Mohammed Darras}
\address{[M. Darras]}
\email{\href{mailto: }{tarifi79@gmail.com}}
\address{}
	
\author[C.\ Fischbacher]{Christoph Fischbacher}
\address{[C.\ Fischbacher] Department of Mathematics, Baylor University, Sid Richardson Bldg., 1410 S.\,4th Street, Waco, TX 76706, USA}
\email{\href{mailto:c_fischbacher@baylor.edu}{ c\_fischbacher@baylor.edu}}

\author[G. Stolz]{G\"unter Stolz}
\address{[G. Stolz] Department of Mathematics, 
University of Alabama at Birmingham, 
Birmingham, AL 35294 USA}
\email{\href{mailto: stolz@uab.edu}{stolz@uab.edu}}

\date{\today}

\begin{abstract}

Schr\"odinger operators with periodic potential have generally been shown to exhibit ballistic transport. In this work, we investigate if the propagation velocity, while positive, can be made arbitrarily small by a suitable choice of the periodic potential. We consider the discrete one-dimensional Schr\"odinger operator $\Delta+\mu V$, where $\Delta$ is the discrete Laplacian, $V$ is a $p$-periodic non-degenerate potential, and $\mu>0$. We establish a  Lieb-Robinson-type bound with a group velocity that scales like $\mathcal{O}(1/\mu)$ as $\mu\rightarrow\infty$. This shows the existence of a linear light cone with a maximum velocity of quantum propagation that is decaying at a rate proportional to $1/\mu$. Furthermore, we prove that the asymptotic velocity, or the average velocity of the time-evolved state, exhibits a decay proportional to $\mathcal{O}(1/\mu^{p-1})$ as $\mu\rightarrow\infty$.

\end{abstract}

\maketitle

%
%

\allowdisplaybreaks

\section{Introduction}
The quantum dynamics of an initially localized state  will generally result in one of the following behaviors, depending on the underlying self-adjoint Hamiltonian: (a) \emph{Localization}, in which the state stays in the initial domain up to an exponentially decaying error for all times, and (b) \emph{delocalization} where the state is expected to spread out as time goes on. Determining whether a quantum system exhibits localization or delocalization, and further exploring bounds on the velocity of this spreading in specific cases, has been a major point of interest from both mathematical and physical perspectives, see, e.g., \cite{ARS23, AG, A-Sabri19,Asch-Knauf-1998, BDMY23, Monval-Sabri-22,  DFO16, DLY,  Darras,  F20,  RS} for single-body systems and 
\cite{
AR23, 
ARFS20, 
ARNSS17, 
ARSS18,
ARSS20, 
LC-Sigal23,
LC-L23,
EKS18, 
EKS18-2, 
FF24, 
FO21, 
FS23, 
GMN,
Imbri16, 
MS18, 
MPS20,
NRSS09,
NS10, 
NS06,
NSS12, 
Seiringer-Warzel16} 
for the many-body setting, just to cite a few.

As a further characterization of delocalized quantum systems, if the average position of the time-evolved state grows linearly in time, then the system is said to exhibit  \emph{ballistic motion}, see \cite{Damanic-Ballistic-24} for a discussion on the notion of ballistic transport. It is generally expected and in some cases rigorously known that periodic quantum systems exhibit non-trivial ballistic quantum transport \cite{Asch-Knauf-1998, Monval-Sabri-22, DLY, F20}. In particular, \cite{Asch-Knauf-1998} proved ballistic transport for continuum  periodic Schr\"odinger operators in any dimension.  \cite{DLY}  established the result for one-dimensional periodic Jacobi operators. Then the ballistic transport result is generalized to Jacobi operators on $\Z^d$ in \cite{F20}, and for more general graphs in \cite{Monval-Sabri-22}.

 In this work  we ask if it is possible to choose a periodic potential $V$ such that the Schr\"odinger operator $\Delta+V$ has arbitrarily small velocity (here $\Delta$ is the discrete Laplacian on $\ell^2(\Z)$). Towards this, we consider the one-dimensional Schr\"odinger operator $\Delta+\mu V$ on $\Z$, $V:\mathbb{\Z}\rightarrow\mathbb{R}$ is a $p$-periodic non-degenerate potential, and $\mu>1$ is considered to be large. We will show that this is related to the velocity of the Schr\"odinger operator $\frac{1}{\mu}\Delta+V$ that corresponds to small hopping terms. A similar scenario is considered in the recent work  \cite{ARS23} on quantum walks in periodic fields. It is shown that the propagation velocity, while positive, can be made arbitrarily small through a careful choice of the periodic field.

We have two main results with the underlying Hamiltonian $\Delta+\mu V$ that is known to exhibit ballistic transport \cite{Asch-Knauf-1998, Monval-Sabri-22, DLY, F20}. Our first result is Thm \ref{thm:main}, it establishes a single-body analog of Lieb-Robinson bounds \cite{LR72} with a group velocity that decays  like $\mathcal{O}(1/\mu)$ as $\mu\rightarrow\infty$. This  is an upper bound on the maximal speed of quantum propagation, and more importantly, it shows that the tail of the evolved state outside of the ballistic motion is  exponentially small. Finding Lieb-Robinson bounds has become a highly active field of research in the context of many-body quantum systems, see e.g.,  \cite{GMN, NS10, NS06} for quantum spin systems, \cite{LR-GGC23, LR-NSY18} for lattice fermions, and \cite{ARSS20, LR-FLS22,NRSS09, NSS12,  LR-YL22} for lattice bosons. Our second main result provides an upper bound for the asymptotic velocity, i.e., the average velocity of the time-evolved state. In particular, Thm \ref{thm:main-2} shows that the asymptotic velocity decays at most like $\mathcal{O}(1/\mu^{p-1})$ as $p\rightarrow\infty$.

Thus, in summary, we prove the existence, and give a general description, of a \emph{linear light cone} (also known as the  \emph{Lieb-Robinson light cone}), see e.g., \cite{LC-Sigal23, LC-L23, LC19, LC15, LC20, LC21}. More precisely, the Lieb-Robinson velocity bound $v_{LR}^{(\mu)}=C/\mu$ in Thm \ref{thm:main} shows that the evolved state decays exponentially outside of a space-time region bounded by the line $t=x/v_{LR}^{(\mu)}$ (where $x$ is the position), which therefore plays the role of a ``narrowing'' light cone at a rate of $1/\mu$ as $\mu\rightarrow\infty$. Then inside the light cone, i.e., within the ballistic motion region, Thm \ref{thm:main-2} shows  that the asymptotic velocity decays at most like $\mathcal{O}(1/\mu^{p-1})$ as $\mu\rightarrow\infty$. Whether the  Lieb-Robinson velocity can be sharpened to be closer to $\mathcal{O}(1/\mu^{p-1})$, and hence whether the corresponding linear light cone can be ``narrowed'' closer to $t\sim \mu^{p-1}v$ is an open question.

We proceed as follows:

In the next section, we provide an overview and a detailed discussion of our model, the underlying problem, and the main results. We introduce essential tools, namely the Floquet transform and the resulting Floquet matrix, which play a pivotal role in establishing bounds for the group velocity. The proof of our first key result, a Lieb-Robinson velocity bound in Thm \ref{thm:main}, is presented in Sections \ref{sec:p=2} and \ref{sec:LR-p}. As a warm-up and to introduce the main idea behind our method, we  address the simpler case of a periodic Hamiltonian with a period of $p=2$ in Section \ref{sec:p=2}. Here all calculations can be done quite explicitly. Then we extend this to general period $p>2$, where we assume non-degenerate potential $V$ in Section \ref{sec:LR-p}. The second main result, concerning a bound for the asymptotic velocity, Thm \ref{thm:main-2}, is discussed in Section \ref{sec:a-v}. Additionally, Appendices \ref{sec:LR-ASY} and \ref{sec:Jacobi} include results that apply to more general situations than those considered in this work. We could not find them in the literature, and believe them to be of independent interest. In particular, Thm \ref{thm:LR-v}, proves that the Lieb-Robinson velocity is an upper bound for the asymptotic velocity for general (single-body) Hamiltonians.  Appendix \ref{sec:Jacobi} provides an explicit formula for the determinant of general Jacobi matrices.

\section{Model, main results, and the Floquet matrix}\label{sec:model-results}

In this section, we present the model, the main results, Theorems \ref{thm:main}, and \ref{thm:main-2}, and show how the analysis reduces to the study of Floquet matrices.

\subsection{Model and main results}
We consider the one-dimensional discrete Schr\"odinger operator with period $p\geq 2$ on $\ell^2(\Z)$ given by
\begin{equation} \label{model}
[H_\mu \psi](n) = [\Delta \psi](n) + \mu V_n \psi(n), \quad n\in \Z.
\end{equation}\label{def:gamma}
Here $[\Delta \psi](n) = \psi(n+1)+\psi(n-1)$ is the discrete Laplacian, $\mu>0$ is a parameter which we will consider large. $V:\Z\rightarrow \mathbb{R}$ is a  periodic potential with period $p$, that is $V_j=V_{np+j}$ for all $n\in\Z$, and we use $V_j:=V(j)$ for $j\in\Z$ to ease the notation. We use $\Gamma$ to denote the range of the periodic potential, i.e.,
\begin{equation}\label{V-Gamma}
 \Gamma:=\max_{0<|j-k|<p}|V_j-V_k|<\infty.
\end{equation}
 Moreover, we require that $V$ is not degenerate, that is, any string of $p$ consecutive potentials is consisting of distinct values. So we set the \emph{gap} $\gamma$ as
\begin{equation}\label{V-gamma}
\gamma:=\min_{0<|j-k|<p}|V_j-V_k|>0 \text{\quad (Non-degeneracy). }
\end{equation}
Our results are stated for $H_\mu$ in (\ref{model}), but they yield implications to more general models of the form
\begin{equation}\label{H-alpha-beta}
H_{\alpha,\beta}:=\alpha\Delta+\beta V, \quad \alpha,\beta\in\mathbb{R}\setminus\{0\}.
\end{equation}
In most of the following, we work with $H_\mu=H_{1,\mu}$, $\mu>0$.

The Hamiltonian $H_\mu$ is understood as a $p\times p$-matrix-valued Jacobi matrix on $\mathcal{H}:=\ell^2(\Z;\C^p)\cong \C^p\otimes \ell^2(\Z)$. We denote the canonical basis of $\C^p$ and $\ell^2(\Z)$ by 
$\{\e_m,\ m=1,\ldots,p\} $ and $\{|j\rangle, j\in\Z\}$,
respectively. We order the basis of $\mathcal{H}$ as
\begin{equation}\label{def:delta}
\delta_{pj+(m-1)}:=|\e_m,j\rangle=|\e_m\rangle\otimes|j\rangle
\end{equation}
for all $m\in\{1,\ldots,p\}$ and $j\in\Z$, e.g., $\delta_0=|\e_1,0\rangle, \delta_1=|\e_2,0\rangle,\ldots, \delta_{p-1}=|\e_p,0\rangle, \delta_{p}=|\e_1,1\rangle,\ldots$.

The Hamiltonian $H_\mu$ is understood in block form (i.e., on $\C^p$ and $\ell^2(\Z)$) as
\begin{equation} \label{model2}
H_\mu= A(\mu)\otimes\idty_{\Z}+ | \e_p\rangle\langle \e_1|\otimes T_{-1}+|\e_1\rangle\langle \e_p|\otimes T_1,
\end{equation}
where $T_{\pm 1}$ are the shift operators by 1 on $\ell^2(\Z)$, i.e., $T_{\pm 1}=\sum_{j\in\Z} |j\pm1\rangle\langle j|$, and $A(\mu)$ is the $p\times p$ matrix
\begin{equation}\label{eq:A}
A(\mu) = \begin{bmatrix} \mu V_1 & 1 & & & \\ 1 & \mu V_2 & \ddots & & \\ & \ddots & \ddots & \ddots & \\ & & \ddots & \mu V_{p-1} & 1 \\ & & & 1 & \mu V_p \end{bmatrix}.
\end{equation}

We are interested in studying velocity bounds for the discrete Schr\"odinger operator $H_\mu$ in the setting of large potential $\mu$. In particular, our bounds quantify the velocity bounds in terms of $\mu$. 

Our main tool is Fourier transform in the second component (Floquet transform) of $\mathcal{H}$, i.e., on $\ell^2(\Z)$. In Section \ref{sec:Fourier}, we discuss this transform in detail. It reduces the problem to the study of the corresponding Floquet matrix. 

Our first main result concerns the norm of the $(j,k)$-th block of $e^{-itH_\mu}\in \C^{p\times p}$ (where $j,k\in\Z$), denoted by $
e^{-itH_\mu}(j,k)$, that reveals a Lieb-Robinson velocity\footnote{We remark here that 
$\displaystyle
e^{-itH_\mu}(j,k)= \Tr_{\ell^2(\Z)}\left[e^{-itH_\mu}(\idty_{\C^p}\otimes |k\rangle\langle j|)\right]
$
, i.e., the partial trace operator of $e^{-itH_\mu}(\idty_{\C^p}\otimes |k\rangle\langle j|)$ by tracing out $\ell^2(\Z)$.}.
\begin{thm} \label{thm:main}
Let $V$ be any non-degenerate $p$-periodic potential with gap $\gamma>0$ and range $\Gamma$ as in (\ref{V-gamma}) and (\ref{V-Gamma}), respectively. 
For any $\eta_0>0$, there are $C_1=C_1(\eta_0,\gamma,\Gamma, p)>0$, $C_2 = C_2(\eta_0,\gamma,\Gamma, p)$ and $\mu_0 = \mu_0(\eta_0,\gamma,\Gamma, p)$ such that
\begin{equation} \label{eq:LR}
\left\|e^{-itH_\mu}(j,k)\right\| \le C_1 e^{-\eta_0( |j-k| -v_{LR}^{(\mu)}|t|)} \text{ where }v_{LR}^{(\mu)}:= \frac{C_2}{\mu}
\end{equation}
for all $t\in \R$, $j,k\in \Z$, and all $\mu \ge \mu_0$. (Here and in the following $\|\cdot\|$ denotes the operator norm.)
\end{thm}
We note here that an explicit numerical value for $C_2$ in (\ref{thm:main}) is given by the following formula, see (\ref{def:C1-C2}),
\begin{equation}
C_2=\frac{\hat C}{\eta_0(\gamma/2)^{p-1}}
\end{equation}
where the value of $\hat C$ can be chosen to be (as a consequence of (\ref{lambdarestrict}) and Lemma \ref{lem:h-p} below),
\begin{equation}\label{def:hat-C}
\hat C:=2^p(\Gamma+\gamma/2)^p+(\Gamma+\gamma/2+2)^{p-2}+2 e^{\eta_0}+\frac{1}{2e^{\eta_0}}.
\end{equation}
Thm \ref{thm:main} represents a Schr\"odinger operator version of the Lieb-Robinson bounds frequently used in the study of quantum many-body systems, see e.g.  \cite{ARSS20, LR-FLS22, GMN, LR-GGC23,  LR72, NS10, NS06, LR-NSY18,  LR-YL22}. This is understood on $\C^p\otimes\ell^2(\Z)$ where the  position of the particle is described by $\ell^2(\Z)$. For any position $j\in\Z$, define the operator $\idty_{|\cdot,j\rangle}$ to be the projection onto the $j$-th position, i.e., onto the space
$
|\cdot,j\rangle:=\Span\{|\e_m,j\rangle, m\in\{1,2,\ldots,p\}\}.
$
For any bounded operators  $A, B$ on $\mathcal{H}$, and positions $j,k\in\Z$ with $j\neq k$, the operators $A_j=\idty_{|\cdot,j\rangle} A \idty_{|\cdot,j\rangle}$ and $B_k=\idty_{|\cdot,k\rangle} B \idty_{|\cdot,k\rangle}$ are ``local'' on sites $j$ and $k$, respectively.  That is $[A_j, B_k]=0$ where $[A_j,B_k]:=A_jB_k-B_kA_j$ is the commutator. Thus, the analogue of Lieb-Robinson bounds in our case corresponds to the commutator $\left[\tau_t^{H_\mu}(A_j), B_k\right]$, where $\tau_t^{H_\mu}(A_j)$ is the Heisenberg evolution 
$
\tau_t^{H_\mu}(A_j):=e^{itH_\mu} A_j e^{-it H_\mu}.
$
Then observe that
\begin{eqnarray}
\left\|\left[\tau_t^{H_\mu}(A_j), B_k\right]\right\|
&\leq& \left\|A_j e^{-it H_\mu}B_k \right\|+ \left\|B_k e^{itH_\mu} A_j\right\| \notag\\
&=& \left\|A_j e^{-it H_\mu}B_k \right\| + \left\|A_j^* e^{-it H_\mu}B_k^* \right\|\notag\\
&\leq& 2\|A_j\| \|B_k\| \left\|e^{-it H_\mu}(j,k)\right\|.
 \end{eqnarray}
This implies that (\ref{eq:LR}) reads as a Lieb-Robinson bound.
\begin{equation}
\left\|\left[\tau_t^{H_\mu}(A_j), B_k\right]\right\| \leq  2\|A_j\| \|B_k\| C_1 e^{-\eta_0( |j-k| -v_{LR}^{(\mu)}|t|)} \text{ where }v_{LR}^{(\mu)}:= \frac{C_2}{\mu}.
\end{equation}
We note here that this argument applies also to general bounded operators with disjoint supports. 

Smallness of the Lieb-Robinson velocity $v_{LR}^{(\mu)}=\mathcal{O}(1/\mu)$ is an effect exclusively due to the potential (i.e., the local part) of the Hamiltonian, in this case due to the choice of large coupling $\mu$. This is of course related to the small hopping setting, $H_{\lambda,1}=\lambda \Delta+V$ where $\lambda$ is small as follows. 
Set $\lambda=1/\mu$ in (\ref{model}) and observe that $H_\mu=H_{1,\mu}=\mu H_{\lambda,1}$, then replace $t$ in (\ref{eq:LR}) by $\lambda t=t/\mu$ to see that (\ref{eq:LR}) reads as
\begin{equation} \label{eq:LR-lambda}
\left\|e^{-it H_{\lambda,1}}(j,k)\right\| \le C_1 e^{-\eta_0( |j-k| - v_{LR}^{(\lambda,1)}|t|)} \text{ where } v_{LR}^{(\lambda,1)}:= C_2 \lambda^2,
\end{equation}
for all $t\in \R$, $j,k\in \Z$, and $\lambda\leq \lambda_0=1/\mu_0$. Here and in the following $v_{LR}^{(\alpha,\beta)}$ denotes the Lieb-Robinson velocity for $H_{\alpha,\beta}$ with the special case $v_{LR}^{(\mu)}:=v_{LR}^{(1,\mu)}$.

Indeed,  if $v$ is a Lieb-Robinson velocity of a general Hamiltonian $H$ then $|\alpha| v$ is the Lieb-Robinson velocity of the Hamiltonian $\alpha H$ where $\alpha\in\mathbb{R}$.  That is, for $H_{\alpha,\mu}=\alpha \Delta+\mu V$, for $\alpha\in\mathbb{R}\setminus\{0\}$ and $\mu>0$ and sufficiently large, we have
\begin{equation}
H_{\alpha,\mu}=\alpha H_{1,\mu/\alpha}\ \Rightarrow\  v_{LR}^{(\alpha,\mu)}=|\alpha|v_{LR}^{(1,\frac\mu\alpha)}=C_2\alpha^2/\mu.
\end{equation}
Here we used the additional fact that $v_{LR}^{(1,\mu)}=v_{LR}^{(1,-\mu)}$ as the negative can be absorbed by the potential.

Thm \ref{thm:main} implies that we have dynamical localization for  times $t$ outside the so-called \emph{Lieb-Robinson light cone} in the space-time region bounded by  $|t|= |j-k|/v_{LB}^{(\mu)}$, see e.g., \cite{LC-L23, LC-Sigal23, LC19, LC15, LC20, LC21}. More precisely, for any $\epsilon\in(0,1)$
\begin{equation}
 \text{ if }\quad v_{LR}^{(\mu)}|t| \le (1-\epsilon)|j-k|
\quad \text{ then }\quad \| e^{-itH_\mu}(j,k) \| \le C_1 e^{-\eta_0 \epsilon|j-k|}.
\end{equation}
This means that $v_{LR}^{(\mu)}= C_2/\mu$ is an 
upper bound for the maximum quantum propagation velocity up to exponentially small tails of the quantum evolution. Said differently, the tail of the time-evolved state outside the region of ballistic motion, with velocity scaling at most like $\mathcal{O}(1/\mu)$, is exponentially small.

Moreover, (\ref{eq:LR}) implies that for any $v>0$,
\begin{eqnarray}
    \sum_{k:\ |j-k|>v|t|}\|e^{-itH_\mu}(j,k)\|^2&\leq&C_1^2 e^{-2\eta_0|t|(v-v_{LR}^{(\mu)})} \sum_{k:\ |j-k|>v|t|}e^{-2\eta_0(|j-k|-v|t|)}\notag\\
    &=& C_1^2\ C(\eta_0)\ e^{-2\eta_0|t|(v-v_{LR}^{(\mu)})}.
\end{eqnarray}
That is, the probability that the distance traveled grows faster than $v_{LR}^{(\mu)}$ decays exponentially, see \cite{AW12}.

Consequently, one should expect to see that $v_{LR}^{(\mu)}=C_2/\mu$ as an upper bound for the asymptotic velocity. More precisely, let $X$ on $\mathcal{H}$ be the position operator
 \begin{equation}\label{def:X}
 X:=\idty_{\C^p}\otimes\sum_{j\in \Z} |j|\ |j\rangle\langle j|, \text{ i.e., }X|\e_m,j\rangle=|j|\ |\e_m,j\rangle \text{ for all }j\in\Z,\ m=1,2,\ldots,p,
 \end{equation}
 then Thm \ref{thm:LR-v}  confirms that (\ref{eq:LR}) implies
 \begin{equation}
v_{Asy}^{(\mu)}:=\lim_{t\rightarrow\infty}\frac1t\left\|X e^{-itH_\mu} \delta_0\right\|\leq v_{LR}^{(\mu)}.
 \end{equation}
 That is, the smallness of the propagation velocity for large potential $\mu$ holds for the expected value of the average velocity. We remark here  that the well known ballistic transport for $H_\mu$, see e.g., \cite{Asch-Knauf-1998, DLY}, asserts  that  $v_{Asy}^{(\mu)}>0$.

 Thm \ref{thm:main} is proven for the special case $p=2$ in Section \ref{sec:p=2}, then proven for general $p>2$ in Section \ref{sec:LR-p}. 
  
 Our second main result is directly related to the asymptotic velocity and it provides bounds when $\mu$ is sufficiently large i.e., $\mu\geq \mu_0$ as in Thm  \ref{thm:main}. In fact, Thm \ref{thm:main-2} shows a faster decay rate for the velocity in $\mu$.
 \begin{thm}\label{thm:main-2}
Given any $\gamma$ as in (\ref{def:gamma}) and period $p\geq 2$. The asymptotic velocity of the $p$-periodic $H_\mu$ given in (\ref{model}) is bounded as
\begin{equation}
v_{Asy}^{(\mu)}=\lim_{t\rightarrow\infty}\frac1t\left\|X e^{-itH_\mu} \delta_0\right\| \leq \frac{C_3}{\mu^{p-1}} \text{ where }\ C_3:=4\pi p^2\left(\frac{2}{\gamma}\right)^{p-1}
\end{equation}
for all $\mu\geq \mu_0$, where $\mu_0$ is the same as in Thm \ref{thm:main}.
\end{thm}
That is, in the large potential setting, i.e., if $\mu$ is sufficiently large, the ballistic motion scales at most like $\mathcal{O}(1/\mu^{p-1})$ for any fixed $p$ and $\gamma$. The proof of Thm \ref{thm:main-2} uses the fact that the asymptotic velocity problem of the $p$-periodic Hamiltonian $H_\mu$ can be approached by studying the derivatives of the eigenvalues of the corresponding $p\times p$ Floquet matrix, see e.g., \cite[Thm 1.6]{DLY}, \cite[Thm 2.3]{Asch-Knauf-1998}, and \cite[Thm 3.1]{Monval-Sabri-22}. Thm \ref{thm:der-Eig-2} provides an explicit formula for these eigenvalues' derivative.

As discussed above for the Lieb-Robinson velocity, we observe that in the case of small hopping terms, $H_{\lambda,1}=\lambda \Delta+V$ where $\lambda$ is sufficiently small, the asymptotic velocity follows trivially from Thm \ref{thm:main-2} as follows.  Set $\lambda=1/\mu$ and note that $H_\mu=\mu H_{\lambda,1}$, to see that
\begin{equation}
\lim_{t\rightarrow\infty}\frac{1}{ t/\mu}\left\|X e^{-i (t/\mu) H_\mu} \delta_0\right\| \leq \frac{C_3}{\mu^{p-1}}=C_3\lambda^{p-1} \Rightarrow \lim_{ t\rightarrow\infty}\frac{1}{t}\left\|X e^{-i t H_{\lambda,1}} \delta_0\right\| \leq C_3\lambda^{p}.
\end{equation}
So the asymptotic velocity scales at most like $\mathcal{O}(\lambda^{p})$ for sufficiently small hopping terms, i.e., when $\lambda<\lambda_0=1/\mu_0$.

Theorems \ref{thm:main} and \ref{thm:main-2} are valid for all $\mu\geq \mu_0$ where $\mu_0(\eta_0,\gamma,\Gamma, p)$ is large enough to ensure that the eigenvalues of the Floquet matrix (see (\ref{def:Floquet-p}) below)  associated with $H_\mu$ are distinct. In particular,  the value of $\mu_0$ is determined by Lemma \ref{lem:lambda-eign} below where $\mu_0$ can be chosen to be $\mu_0=1/\lambda_0$. A definite numeric value for $\mu_0$ (or $\lambda_0$) is given as
\begin{equation}\label{eq:mu-0}
\mu_0^2=\max\left\{1,2\hat{C}(2/\gamma)^p\right\} 
\end{equation}
where $\hat{C}=\hat{C}(\eta_0,\gamma,\Gamma, p)>1$ is given in (\ref{def:hat-C}).
 An interesting scenario to consider is when $\gamma\ll 1$, then $\mu$ has to be taken as $\mu\geq \mu_0\gg (2/\gamma)^{p/2}$. This clarifies that  Theorems \ref{thm:main} and \ref{thm:main-2} do not allow to consider the case when $p$ and $\mu$ are constants and $\gamma\rightarrow 0$. However,  for any fixed  $p$ and  $\gamma$, the group velocity $v_{LR}^{(\mu)}\sim\mathcal{O}(1/\mu)$ and the asymptotic velocity $v_{Asy}^{(\mu)}\sim \mathcal{O}(1/\mu^{p-1})$ as $\mu\rightarrow\infty$. The question remains unanswered regarding the possibility of enhancing the Lieb-Robinson velocity to the scale closer to $\mathcal{O}(1/\mu^{p-1})$. Consequently, it is yet to be determined whether the corresponding linear light cone can be narrowed down closer to $t\sim \mu^{p-1}|j-k|$.

\subsection{Fourier transform and the Floquet matrix}\label{sec:Fourier}

An explicit implementation of the Floquet transform of $H_\mu$ is given by the componentwise Fourier transform $\F: \ell^2(\Z;\C^p) \to \mathcal{L}^2([0,2\pi);\C^p)$,\footnote{In the block Jacobi form of $H_\mu$ in (\ref{model2}), the Floquet transform  is understood as $\F: \C^p\otimes\ell^2(\Z)\rightarrow \mathcal{L}^2([0,2\pi))^{\oplus p}$, given by $\F=\idty_{\C^p}\otimes F$ where $(Fu)(x)=(2\pi)^{-1/2}\sum_{k} e^{ixk} u(k)$}
\begin{equation} \label{eq:F}
(\F u)(x) = (2\pi)^{-1/2} \sum_n \begin{bmatrix} u_1(n) \\ \vdots\\ u_p(n) \end{bmatrix} e^{ixn}, \quad (\F^{-1}\varphi)(n) = (2\pi)^{-1/2} \int_0^{2\pi} \begin{bmatrix} \varphi_1(x) \\\vdots\\ \varphi_p(x) \end{bmatrix} e^{-ixn}\,dx.
\end{equation}
We get 
\begin{equation} \label{FtransH} 
H_\mu = \F^{-1} \hat{H}_\mu \F
\end{equation} 
with $\hat{H}_\mu$ on $\mathcal{L}^2([0,2\pi);\C^p)$ given by the generalized multiplication operator
\begin{equation}
(\hat{H}_\mu \varphi)(x) = A_p(\mu, e^{ix}) \varphi(x), \quad x \in [0,2\pi).
\end{equation}
Here $\varphi = \begin{bmatrix}\varphi_1&\ldots&\varphi_p\end{bmatrix}^T \in \mathcal{L}^2([0,2\pi); \C^p)$, and for $p>2$, the Floquet matrix $A_p(\mu, e^{ix})$ is given as 
\begin{equation}\label{def:Floquet-p}
A_p(\mu, e^{ix}) = \begin{bmatrix} \mu V_1 & 1 & & & e^{ix} \\ 1 & \mu V_2 & \ddots & & \\ & \ddots & \ddots & \ddots & \\ & & \ddots & \mu V_{p-1} & 1 \\ e^{-ix} & & & 1 & \mu V_p \end{bmatrix}_{p\times p}.
\end{equation}
Using \eqref{FtransH} we find that the $(j,k)$-th $p\times p$-block of $e^{-itH_\mu}$ is given by the formula
\begin{equation}\label{eq:block-exp-jk}
e^{-itH_\mu}(j,k) = \frac{1}{2\pi} \int_0^{2\pi} e^{-itA_p(\mu,e^{ix})} e^{-i(j-k)x}\,dx.
\end{equation}
When $p=2$, then $A_{p=2}(\mu, e^{ix})$ is the $2\times 2$ matrix
\begin{equation}\label{def:Floquet-2}
A_2(\mu, e^{ix})=\begin{bmatrix} \mu V_1 & 1+e^{ix}\\ 1+e^{-ix} & \mu V_2
\end{bmatrix}
\end{equation}
for which explicit exact calculations can be made to prove the results. The next section describes this case as a warm-up for the proof of the general period  $p$. Note also that $A_2(\mu,e^{ix})$ has a different structure than $A_p(\mu,e^{ix})$ when $p>2$.

\section{Warm-up: Lieb-Robinson velocity bound for period two} \label{sec:p=2}
To put emphasis on the method, and as a warm up with explicit calculations, we consider in this section the periodic Schr\"odinger operator $H_\mu$ in (\ref{model}) with alternating potential $V_n:=V(n)=(-1)^n$ for all $n\in\Z$ (so that $\gamma=2$), i.e., we consider the Hamiltonian $H^{(2)}_\mu$   on $\ell^2(\Z)$ given as
\begin{equation}\label{model-main-p=2}
[H^{(2)}_{\mu}\psi](n)=[\Delta \psi](n)+\mu (-1)^n \psi(n), \quad n\in\Z
\end{equation}
where $\Delta$ is the discrete Laplacian on $\ell^2(\Z)$.

In this case, the Floquet matrix (\ref{def:Floquet-2}) is given by
\begin{equation}\label{eq:A-2}
A_2(\mu, e^{ix})=\begin{bmatrix} \mu & 1+e^{ix}\\ 1+e^{-ix} & -\mu\end{bmatrix}
=
\mathscr{O} \begin{bmatrix} \omega_+(x) & 0 \\ 0 & \omega_-(x) \end{bmatrix} \mathscr{O}^{-1}.
\end{equation}
Here, for $x \not= \pi$, the eigenvalues $\omega_{\pm}(x)$ of $A_2(\mu, e^{ix})$ are given by the formula
\begin{equation}\label{eq:Floquet-Eig-2by2}
 \omega_\pm(x) = \pm \sqrt{\mu^2+ 2(1+\cos x)},
 \end{equation}
 and 
\begin{equation}\label{eq:A2-O}
 \mathscr{O}= \begin{bmatrix} 1 & 1 \\ q_+ & q_- \end{bmatrix}, \quad q_{\pm} = \frac{\omega_\pm(x)-\mu}{1+e^{ix}},\quad \mathscr{O}^{-1} = - \frac{1+e^{ix}}{2\omega_+(x)} \begin{bmatrix} q_- & -1 \\ -q_+ & 1 \end{bmatrix}.
 \end{equation}
We note that this implies that the spectrum of $H^{(2)}_\mu$ is 
\begin{equation}
\sigma(H^{(2)}_\mu) = \bigcup_{x\in [0,2\pi)} \{\omega_{\pm}(x)\} = \left[-\sqrt{\mu^2+4}, -\mu\right] \cup \left[ \mu, \sqrt{\mu^2+4}\right].
\end{equation}
 In particular, this gives that (the Lebesgue measure of the spectrum) $|\sigma(H^{(2)}_\mu)| \le 4/\mu$ (a special case of a result by Last in \cite{L}).

However, we are interested in understanding the quantum dynamics associated with $H^{(2)}_\mu$. Towards this, with direct calculations using (\ref{eq:A-2}) and (\ref{eq:A2-O}), we find that the entries of the $2\times 2$ matrix $e^{-itA_2(\mu,e^{ix})}$ are given by the formulas
\begin{eqnarray} \label{evolution}
\langle \e_1,e^{-itA_2(\mu,e^{ix})}\e_1\rangle& = & \frac12(f_{t,\mu}^-+f_{t,\mu}^+)(x)+\frac12 (g_{t,\mu}^- -g_{t,\mu}^+)(x)\notag
 \\
\langle \e_2,e^{-itA_2(\mu,e^{ix})}\e_2\rangle &=&
 \overline{\langle \e_1,e^{-itA_2(\mu,e^{ix})}\e_1\rangle}\notag\\
\langle \e_1,e^{-itA_2(\mu,e^{ix})}\e_2\rangle& = &\frac{1+e^{ix}}{2\mu}(g_{t,\mu}^+ - g_{t,\mu}^-)(x)\notag\\
\langle \e_2,e^{-itA_2(\mu,e^{ix})}\e_1\rangle& = & -\overline{\langle \e_1,e^{-itA_2(\mu,e^{ix})}\e_2\rangle}
\end{eqnarray}
where 
\begin{equation}\label{def:f-g}
f_{t,\mu}^{\pm}(x) := e^{\pm it \omega_+(x)}\quad \text{ and }\quad g_{t,\mu}^{\pm}(x) := \frac{\mu}{\omega_+(x)}f_{t,\mu}^{\pm}(x)=\frac{\mu e^{\pm it\omega_+(x)}}{\omega_+(x)}. 
\end{equation}
The trivial case $x=\pi$, where both $A_2(\mu,e^{ix})$ and $e^{-itA_2(\mu,e^{ix})}$ are diagonal, appears as a limiting case of \eqref{evolution}. 
Here recall that the $2\times 2$ matrix $e^{-itH^{(2)}_\mu}(j,k)$ is given by formula (\ref{eq:block-exp-jk}). We will prove the following  result.

\begin{thm} \label{thm:p2}
There exist constants $C_1>0$, $C_2>0$ and $\eta>0$ such that
\begin{equation} \label{LRbound}
\| e^{-itH^{(2)}_\mu}(j,k) \| \le C_1 e^{-\eta(|j-k|- \frac{C_2}{\mu} |t|)} 
\end{equation}
for all $t\in \R$, $j, k \in \Z$ and $\mu\ge 1$. 
\end{thm}

\begin{proof}
Proving Thm \ref{thm:p2} reduces to showing corresponding bounds for the Fourier coefficients of the matrix entries appearing in \eqref{evolution}. In particular, we use (\ref{eq:block-exp-jk}) to bound each of the entries of the matrix $e^{-itH^{(2)}_\mu}(j,k)$
\begin{equation}\label{eq:entries-exp-p2}
\langle \e_\ell, e^{-itH^{(2)}_\mu}(j,k)\e_m\rangle= \frac{1}{2\pi} \int_0^{2\pi} \langle \e_\ell, e^{-itA_2(\mu,e^{ix})} \e_m\rangle e^{-i(j-k)x}\,dx
\end{equation}
for $\ell,m\in\{1,2\}$.
To handle the entries of $e^{-itH^{(2)}_\mu}(j,k)$ we have to prove bounds of the form (\ref{LRbound}) for the entries given in (\ref{evolution}), i.e., we need to find such bounds for the functions
$f_{t,\mu}^{\pm}(x)$ and $g_{t,\mu}^{\pm}(x)$ in (\ref{def:f-g}). More precisely, with the notation
\begin{equation}\label{eq:p2-11-term}
\hat{\phi}(k) := \frac{1}{2\pi} \int_0^{2\pi} \phi(x) e^{-ikx}\,dx,
\end{equation}
we will show that 
\begin{equation} \label{hat-f-g-bound}
|\hat{f}^\pm_{t,\mu}(k)| \le C_2 e^{-\eta(|k|- \frac{C'_2}{\mu} |t|)} \quad\text{ and }\quad |\hat{g}_{t,\mu}^\pm(k)| \le C_3 e^{-\eta(|k|- \frac{C'_2}{\mu} |t|)}.
\end{equation}
Then use these bounds  in (\ref{eq:entries-exp-p2}) using (\ref{evolution}) to find that 
\begin{eqnarray}
|\langle \e_1, e^{-itH^{(2)}_\mu}(j,k) \e_2\rangle| &\leq& 2C_3 e^{-\eta(|k|- \frac{C_2'}{\mu} |t|)} \notag \\  
|\langle \e_\ell, e^{-itH^{(2)}_\mu}(j,k) \e_\ell\rangle| &\leq& (C_2+C_3) e^{-\eta(|k|- \frac{C_2'}{\mu} |t|)}\ \text{ for } \ \ell=1,2.
\end{eqnarray}
In the following, we show the bounds (\ref{hat-f-g-bound}).

Since $\mu \ge 1$ there is a (fixed) $\rho_0 >1$ such that
\begin{equation} \label{apriori}
\mu^2+2 > 2\rho_0 + \frac{1}{2\rho_0}.
\end{equation}
\begin{remark}
We note that  the value of $\rho_0$ in  (\ref{apriori})  can be chosen independently of $\mu\geq 1$. In particular, it can be any value in the open interval $\left(1,\frac34+\frac{\sqrt{5}}{4}\right)$. This can be seen for example, by observing that the function $2\rho_0+(2\rho_0)^{-1}$ is increasing on the interval  $\left(1,\frac34+\frac{\sqrt{5}}{4}\right)$, and it attains the value of $3<\mu^2+2$ at the endpoint $\frac{3}{4}+\frac{\sqrt{5}}{4}$.
\end{remark}
First, observe that 
\begin{equation}\label{eq:hat-f-z1}
\hat{f}^\pm_{t,\mu}(k) = \frac{1}{2\pi} \int_0^{2\pi} f^\pm_{t,\mu}(x) e^{-ikx}\,dx= \frac{1}{2\pi i} \oint_{|z|=1} e^{\pm it h_\mu(z)} z^{-k-1}\,dz
\end{equation}
where we applied the change of variable $e^{ix}=z$ (so $dx=dz/(zi)$) and the function $h_\mu$ is given as
\begin{equation}
h_\mu(z) := \sqrt{\mu^2+2+z+1/z}.
\end{equation}
The contour integral in (\ref{eq:hat-f-z1}) is taken over the positively oriented unit circle $|z|=1$.
By \eqref{apriori}, $h_\mu(z)$ is analytic in the annulus $\frac{1}{2\rho_0} < |z| < 2\rho_0$ since
\begin{equation}\label{z+z-1}
\left|z+1/z\right| \le |z| + 1/|z| < 2\rho_0 + \frac{1}{2\rho_0} < \mu^2 + 2.
\end{equation}
Here we used the fact that function $x+1/x$ is increasing for $x>1$.

We now treat the two cases $k\geq 0$ and $k<0$ separately. 
If $k\geq 0$, we proceed as follows. We replace the integration over the unit circle $|z|=1$  in \eqref{eq:hat-f-z1} by the integration over the larger circle $|z|=\rho_0$, so that
\begin{equation}
\hat{f}^\pm_{t,\mu}(k) = \frac{1}{2\pi i} \oint_{|z|=\rho_0} e^{\pm ith_\mu(z)} z^{-k-1}\,dz.
\end{equation}
This gives the bound
\begin{equation}
|\hat{f}^\pm_{t,\mu}(k)| \le \left( \max_{|z|=\rho_0} e^{ t| \IM(h_\mu(z))|} \right) \rho_0^{-k}.
\end{equation}
For $h_\mu(z)$ we have, by first order Taylor expansion in the annulus $\frac{1}{2\rho_0} < |z| < 2\rho_0$, that \footnote{Here, in the context of complex-valued functions, the big-o notation $\mathcal{O}(\cdot)$ is understood as follows: Let $f$ and $g$ be complex valued functions defined on a set $\mathscr{S}\subset\C$, the formula  $f(z)=\mathcal{O}(g(z))$ means that there is a constant $C>0$ such that
$|f(z)|\leq C|g(z)|$  for all $z\in \mathscr{S}$, see e.g., \cite{Bruijn-Complex}.}
(here we recall (\ref{z+z-1})) 
\begin{eqnarray} \label{hmubound}
h_\mu(z) & = & \sqrt{\mu^2+2 + z +1/z}=\sqrt{\mu^2+2}\left(1+\mathcal{O}\left(\frac{z+1/z}{\mu^2+2}\right)\right) \notag \\
& = & \sqrt{\mu^2+2} +  \mathcal{O}\left(\frac{z+1/z}{\sqrt{\mu^2+2}} \right) = \sqrt{\mu^2+2} + \mathcal{O}\left(\frac{2\rho_0+\frac{1}{2\rho_0}}{\sqrt{\mu^2+2}} \right). 
\end{eqnarray}

Considering that $\IM(\sqrt{\mu^2+2})=0$, we obtain
\begin{equation}
|\IM(h_\mu(z))| \le C_0 \frac{2\rho_0 + \frac{1}{2\rho_0}}{\sqrt{\mu^2+2}} \le \frac{C'}{\mu}.
\end{equation}
Hence,
\begin{equation} \label{hatf1}
|\hat{f}^\pm_{t,\mu}(k)| \le e^{-\eta k + \frac{C'}{\mu} t}, \quad \eta := \ln\rho_0.
\end{equation}

For the case $k\le -1$ we deform the integration contour $|z|=1$ in \eqref{eq:hat-f-z1} to the smaller circle $|z|= 1/\rho_0$. We get
\begin{equation}
\hat{f}^\pm_{t,\mu}(k) = \frac{1}{2\pi i} \oint_{|z|=1/\rho_0} e^{\pm ith_\mu(z)} z^{-k-1}\,dz
\end{equation}
with
\begin{equation} \label{hatf2}
|\hat{f}^\pm_{t,\mu}(k)| \le  \left( \max_{|z|=1/\rho_0} e^{t|\IM(h_\mu(z))|} \right) (1/\rho_0)^{-k}  
 \le  e^{\frac{C'}{\mu} t} \rho_0^{-|k|} = e^{-\eta |k| + \frac{C'}{\mu} t}. 
\end{equation}
Combining the bounds \eqref{hatf1} and \eqref{hatf2}, for the two cases $k\ge 0$ and $k<0$, proves the bound for $|\hat{f}^\pm_{t,\mu}(k)|$ in  \eqref{hat-f-g-bound} with $C_2=1$, $C'_2=C'/\eta$.

By \eqref{hmubound} we have 
\begin{equation}
\frac{\mu}{\sqrt{\mu^2+2+z + 1/z}} = \frac{\mu}{\sqrt{\mu^2+2} + \mathcal{O}\left(\frac{2\rho_0 + \frac{1}{2\rho_0}}{\sqrt{\mu^2+2}}\right)} \Rightarrow  
\left|\frac{\mu}{\sqrt{\mu^2+2+z + 1/z}} \right|\le C_3,
\end{equation}
both on the circles $|z|= \rho_0$ and $|z|=1/\rho_0$. With this uniform bound we can proceed as before and get the bound for $|\hat{g}^\pm_{t,\mu}(x)|=\left|\frac{\mu}{\omega_+(x)}\right| |f_{t,\mu}^{\pm}(x)|$ in \eqref{hat-f-g-bound}. This finishes the proof of Thm \ref{thm:p2}.

\end{proof}

\section{A Lieb-Robinson velocity for general period}\label{sec:LR-p}

Now we consider the Schr\"odinger operator $H_\mu$ with periodic potential $V:\Z \to \R$ and general period $p> 2$ as described in Section \ref{sec:model-results}  from (\ref{model}) to (\ref{eq:A}). This section includes the proof of Thm \ref{thm:main}.

Section \ref{sec:Fourier} presents the Fourier transform that yields the Floquet matrix $A_p(\mu,e^{ix})$ in (\ref{def:Floquet-p}) and formula (\ref{eq:block-exp-jk}) for the $p\times p$ matrix $e^{-itH_\mu}(j,k)$. We recall (\ref{eq:block-exp-jk}) for the reader's convenience
\begin{equation}\label{eq:block-exp-jk-2}
e^{-itH_\mu}(j,k) = \frac{1}{2\pi} \int_0^{2\pi} e^{-itA_p(\mu,e^{ix})} e^{-i(j-k)x}\,dx.
\end{equation}

Instead of the large parameter $\mu$ it will be more convenient to mostly work with the small parameter $\lambda := 1/\mu$ and write $A_p(\mu,z) = \mu \tilde{A}_p(\lambda,z)$ with
\begin{equation}\label{def:floquet-m}
\tilde{A}_p(\lambda, z) = \begin{bmatrix} V_1 & \lambda & & & \lambda z \\ \lambda & V_2 & \ddots & & \\ & \ddots & \ddots & \ddots & \\ & & \ddots & V_{p-1} & \lambda \\ \lambda/z & & & \lambda & V_p \end{bmatrix}.
\end{equation}
Here we do not only consider the case $z=e^{ix}$, but also extensions of $\tilde{A}_p(\lambda,z)$ into rings $1/\rho_0 < |z| < \rho_0$ in $\C$ for $\rho_0>1$. Note that for $|z|\not= 1$ the resulting operators are {\it not} self-adjoint (where $\lambda\neq 0$). It is noteworthy to remark here that $\tilde{A}_p(\lambda, e^{ix})$ in (\ref{def:floquet-m}) is the Floquet matrix of $H_{\lambda,1}=H_{\frac1\mu,1}=\frac{1}{\mu} \Delta+V$.

The eigenvalues $V_1, \ldots, V_p$ of $\tilde{A}_p(0,z)$ are real and pairwise distinct. The following Lemma  describes a parameter region $(\lambda,z)$ in which the eigenvalues of $\tilde{A}_p(\lambda,z)$ remain distinct and thus are analytic both with respect to $\lambda$ and with respect to $z$.

\begin{lem}\label{lem:lambda-eign}
For any fixed $\rho_0>1$, $z\in\C$ in the annulus $(2\rho_0)^{-1}<|z|<2\rho_0$, and  $\gamma>0$ as in (\ref{def:gamma}). There exists $\hat{C}=\hat{C}(\rho_0, \gamma,\Gamma,p)>1$ such that the eigenvalues of $\tilde{A}_p(\lambda,z)$ can be labeled as  $\{\zeta_\ell(\lambda,z),\ \ell=1,2,\ldots,p\}$ and
\begin{equation}\label{lem:eq-lambda-eign}
|\zeta_\ell(\lambda,z)-V_\ell|\leq \frac{\lambda^2\hat{C}}{(\gamma/2)^{p-1}}<\gamma/4
\end{equation} 
for all $\ell\in\{1,\ldots,p\}$, and
 for all $\lambda\in[0, \lambda_0]$, where $\lambda_0^2=\min\{1,\frac{(\gamma/2)^p}{2 \hat{C}}\}$.
 \end{lem}
The proof of Lemma \ref{lem:lambda-eign} is included at the end of this section.
\begin{remark}
We note here that when $z=e^{ix}$, an asymptotic formula for each of the eigenvalues $\zeta_\ell(\lambda,e^{ix})\in\R$  is understood from \cite[pages 6-8]{Reed-Simon-IV} (see also \cite[Thm 2.2]{KS-Eig14}).
\begin{equation}
\zeta_\ell(\lambda,e^{ix})=V_\ell-\lambda^2\sum_{j=1,\ j\neq \ell}^p\frac{|\langle\e_j, \tilde{A}_p(\lambda, e^{ix})\e_\ell\rangle|^2}{V_j-V_\ell}+\mathcal{O}(\lambda^3) \text{ as }\lambda\rightarrow 0.
\end{equation}
i.e., $|\zeta_\ell(\lambda,e^{ix})-V_\ell |=\mathcal{O}(\lambda^2)$. This shows that our bound in (\ref{lem:eq-lambda-eign}) is sharp in the order of $\lambda$.

\end{remark} 
 
That is, for sufficiently small $\lambda\leq \lambda_0<1$ (or equivalently, sufficiently large $\mu=1/\lambda>1$), each eigenvalue $\zeta_\ell(\lambda,z)$ of $\tilde{A}_p(\lambda,z)$ has distance of at most $\gamma/4$ to $V_\ell$ for all $z\in\C$ such that $(2\rho_0)^{-1}<|z|<2\rho_0$.  Hence, $\{\zeta_\ell(\lambda,z)\}_\ell$ are distinct and analytic with respect to $\lambda$ and $z$ on the window
\begin{equation}\label{eq:range-lambda-z}
(\lambda,z)\in [0,\lambda_0]\times\left\{z\in\C;\ |z|\in\big[(2\rho_0)^{-1},2\rho_0\big]\right\}.
\end{equation}
This parameter range is considered in the sequel of this section.
 We stress here that $\mu_0$ in Thm \ref{thm:main} is determined by the value of $\lambda_0$ where $\mu_0$ can be chosen to be $\mu_0=1/\lambda_0$.

Lemma \ref{lem:lambda-eign} shows that within the parameter range \eqref{eq:range-lambda-z} we have
\begin{equation} \label{IMbound}
\left| \frac{1}{\lambda}  \IM\left(\zeta_\ell(\lambda,z)\right) \right| \le\frac{\lambda\hat{C}}{(\gamma/2)^{p-1}}, \quad \ell=1,\ldots,p.
\end{equation}
Then we direct our attention to the integral formula (\ref{eq:block-exp-jk-2}) for $e^{-itH_\mu}(j,k)$, and we use the spectral decomposition of $e^{-itA_p(\mu,e^{ix})}$,
\begin{equation} \label{eq2}
e^{-itA_p(\mu,e^{ix})} = \sum_{\ell=1}^p e^{-it\mu \zeta_\ell(\lambda,e^{ix})} P_\ell(\lambda,e^{ix}),
\end{equation}
where $P_\ell(\lambda,e^{ix})$ is the orthogonal rank-one projection onto the eigenspace of $\tilde{A}_p(\lambda,e^{ix})$ associated to the eigenvalue $\zeta_\ell(\lambda, e^{ix})$. By the Riesz formula
\begin{equation} \label{eq3}
P_\ell(\lambda,e^{ix}) = \frac{1}{2\pi i} \oint_{\Gamma_{\gamma/2}(V_\ell)} (\zeta\idty - \tilde{A}_p(\lambda, e^{ix}))^{-1}\, d\zeta,
\end{equation}
with the (positively oriented) circle $\Gamma_{\gamma/2}(V_\ell) = \{ \zeta \in\C;\ |\zeta -V_\ell |= \gamma/2\}$ whose interior contains only the eigenvalue $\zeta_\ell(\lambda,e^{ix})$ for $\ell=1,2,\ldots,p$ and no others. Here and in the following we use  $\idty:=\idty_{\C^p}$ to simplify notations.

Substitute (\ref{eq3}) in (\ref{eq2}) then in (\ref{eq:block-exp-jk-2}) to obtain
\begin{equation}\label{eq:iterative-int-1}
e^{-it H_\mu}(j,k)=\frac{1}{4\pi^2 i} \sum_{\ell=1}^p  \int_{0}^{2\pi}\oint_{\Gamma_{\gamma/2}(V_\ell)}e^{-it\mu\zeta_\ell(\lambda,e^{ix})}(\zeta\idty-\tilde{A}_p(\lambda,e^{ix}))^{-1}e^{-i(j-k)x}\ d\zeta\ dx.
\end{equation}
Note that \eqref{lem:eq-lambda-eign} shows that
\begin{equation}
|\zeta-\zeta_k(\lambda,e^{ix})|\geq \gamma/4, \text{ for all }\quad k=1,\ldots,p, \text{ and }\zeta\in \Gamma_{\gamma/2}(V_\ell).
\end{equation}
 This gives directly that 
\begin{equation}
\|(\zeta\idty-\tilde{A}_p(\lambda,e^{ix}))^{-1}\|\leq 4/\gamma.
\end{equation}
Hence, by Fubini the order of two integrals in (\ref{eq:iterative-int-1}) can be switched. 

We do so, and we  make the change of variables $z=e^{ix}$. The integral over $x$ 
\begin{equation}\label{eq:int-x}
\int_0^{2\pi} e^{-it\mu \zeta_\ell(\lambda,e^{ix})} (\zeta\idty-\tilde{A}_p(\lambda,e^{ix}))^{-1} e^{-i(j-k)x}\,dx
\end{equation}
is then given  as 
\begin{eqnarray} \label{contchange}
(\ref{eq:int-x}) &=&  -i \oint_{|z|=1} e^{-it\mu \zeta_\ell(\lambda,z)} (\zeta\idty-\tilde{A}_p(\lambda,z))^{-1} z^{-(j-k+1)}\,dz \notag \\
& = & -i \oint_{|z|=\rho_0} e^{-it\mu \zeta_\ell(\lambda,z)} (\zeta\idty - \tilde{A}_p(\lambda,z))^{-1} z^{-(j-k+1)}\,dz. 
\end{eqnarray}
In (\ref{contchange}), we have used analyticity in $\frac{1}{2\rho_0} < |z| < 2\rho_0$ to deform the integration contour from $|z|=1$ to $|z|=\rho_0>1$ and, for the moment, only consider the case $j\ge k$.

Next, we will find a Lieb-Robinson velocity bound for the $(n, m)$-matrix elements of $e^{-itH_\mu}(j,k)\in\C^{p\times p}$ given by \eqref{eq:iterative-int-1} and \eqref{contchange}, 
\begin{equation}
|\langle \e_n,e^{-itH_\mu}(j,k) \e_m\rangle|, \text{ for all }n, m\in\{1,2,\ldots,p\}.
\end{equation}
This reduces to the problem of bounding the resolvent 
$\langle \e_n, (\zeta\idty-\tilde{A}_p(\lambda,z))^{-1} \e_m \rangle$ 
which can be done with the aid of  Cramer's Rule 
\begin{equation}
\langle \e_n, (\zeta\idty-\tilde{A}_p(\lambda,z))^{-1} \e_m \rangle=\det(\zeta\idty-\tilde{A}_p(\lambda,z))^{-1} \langle \e_n,\text{adj}(\zeta\idty-\tilde{A}_p(\lambda,z))\e_m\rangle.
\end{equation}
Here an in the following, $\text{adj}(\cdot)$ denotes the \emph{adjugate} matrix.

 Then for $\zeta\in\Gamma_{\gamma/2}(V_\ell)$ we have
\begin{equation}
|\det(\zeta\idty-\tilde{A}_p(\lambda,z))| = \left| \prod_{k=1}^p (\zeta-\zeta_k(\lambda,z)) \right| \ge \left( \frac{\gamma}{4} \right)^p.
\end{equation}
Moreover, since $ \langle \e_n,\text{adj}(\zeta\idty-\tilde{A}_p(\lambda,z))\e_m\rangle$ is a polynomial in $\lambda\in[0, \lambda_0]$, $|z|=\rho_0$, and $\zeta\in\Gamma_{\gamma/2}(V_\ell)$, there is a constant $C_4(\rho_0,\gamma,\Gamma,p)$ such that
\begin{equation}\label{eq:resolvent-bound-0}
| \langle \e_n,\text{adj}(\zeta\idty-\tilde{A}_p(\lambda,z))\e_m\rangle| \le C_4(\rho_0,\gamma, \Gamma,p).
\end{equation}
Hence
\begin{equation}\label{eq:resolvent-bound}
|\langle \e_n, (\zeta\idty-\tilde{A}_p(\lambda,z))^{-1} \e_m \rangle|\leq C_4(\rho_0,\gamma, \Gamma,p)  \left( \frac{\gamma}{4} \right)^{-p}=:\tilde{C}_4(\rho_0,\gamma,\Gamma,p).
\end{equation}
\begin{remark}
Unlike $\hat C$ in Lemma \ref{lem:lambda-eign}, the constant $C_4$ (and hence $\tilde{C}_4$) does not affect the Lieb-Robinson velocity as it appears in front  of the exponential function, see (\ref{12}) below. For this, quantifying (or having an explicit formula for) $C_4$ is not relevant for our work.
\end{remark}
Finally, by \eqref{IMbound} we have
\begin{equation}\label{eq:bound2}
\left |e^{-it\mu \zeta_\ell(\lambda,z)} z^{-(j-k+1)} \right| \le \exp\left(\frac{\lambda \hat{C}}{(\gamma/2)^{p-1}}|t|\right)\, \rho_0^{-(j-k+1)}.
\end{equation}
To wrap up, use the bounds (\ref{eq:bound2}) and (\ref{eq:resolvent-bound}) in (\ref{contchange}) then in (\ref{eq:iterative-int-1}) to obtain
\begin{eqnarray}\label{12}
|\langle \e_\ell,e^{-itH_\mu}(j,k) \e_m\rangle|&\leq&  \frac{1}{4\pi^2}\tilde{C}_4(\rho_0,\gamma,\Gamma,p)\ e^{\frac{\lambda \hat{C}}{(\gamma/2)^{p-1}}|t|}\, \rho_0^{-(j-k+1)}\sum_{\ell=1}^p \oint_{\Gamma_{\gamma/2}(V_\ell)} \oint_{|z|=\rho_0}\ dz\ d\zeta \nonumber\\
&=&\frac{p\gamma}{2}\tilde{C}_4(\rho_0,\gamma,\Gamma,p) \exp\left(\frac{\lambda \hat{C}}{(\gamma/2)^{p-1}}|t|-(\ln \rho_0)|j-k|\right),\quad j>k.
\end{eqnarray}
This gives directly the desired result
\begin{equation}\label{arsbound-1}
\|e^{-itH_\mu}(j,k)\|\leq   \frac{p\gamma}{2}\tilde{C}_4\ e^{-\eta_0( |j-k| - \frac{C_2}{\mu}|t|)}
\end{equation}
where
\begin{equation}\label{def:C1-C2}
 \eta_0:=\ln\rho_0>0, \text{ and } C_2:=\frac{\hat{C}}{\eta_0(\gamma/2)^{p-1}}.
\end{equation}

So far, we have considered $j\ge k$. For $j<k$ one can argue similarly, but now deform the integration contour in \eqref{contchange} from $|z|=1$ to $|z|=1/\rho_0$, leading to a bound of the same form as \eqref{arsbound-1}. In particular, the bound (\ref{eq:resolvent-bound-0}), and hence in (\ref{eq:resolvent-bound}) would be different, i.e., there exists a constant $\tilde{C}_1(\rho_0,\gamma,\Gamma,p)<\infty$ such that
\begin{equation}
|\langle \e_\ell, (\zeta\idty-\tilde{A}_p(\lambda,z))^{-1} \e_m \rangle|\leq \tilde{C}_5(\rho_0,\gamma,\Gamma,p), \text{ for all }\lambda\in[0,\lambda_0], |z|=1/\rho_0. 
\end{equation}
Then
\begin{equation}\label{arsbound-2}
\|e^{-itH_\mu}(j,k)\|\leq  \frac{p\gamma}{2}\tilde{C}_5\ e^{-\eta_0( |j-k| - \frac{C_2}{\mu}|t|)} \end{equation}
In Thm \ref{thm:main},  $\eta_0$ and $C_2$ are as in (\ref{def:C1-C2}), and 
$C_1:=\frac{p\gamma}{2}\max\{\tilde{C}_4,\tilde{C}_5\}$.

It remains to prove Lemma \ref{lem:lambda-eign} to finish the proof of Thm \ref{thm:main}.

\begin{proof}[Proof of Lemma \ref{lem:lambda-eign}]
For a given $\lambda\in[0,1]$, $\rho_0>1$, and $z\in\C$ such that $|z|\in((2\rho_0)^{-1}, 2\rho_0)$,
consider the characteristic polynomial of the Floquet matrix $\tilde{A}_p(\lambda,z)$ in (\ref{def:floquet-m}).
\begin{eqnarray} \label{charpol}
F(\zeta) & = & \det( \zeta\idty-\tilde{A}_p(\lambda,z)) \notag\\
& = & \det
\begin{bmatrix}
 \zeta-V_1 & -\lambda & & \\ -\lambda & \ddots & \ddots & \\ & \ddots & \ddots & -\lambda \\ & & -\lambda & \zeta-V_{p} 
\end{bmatrix}+\\
&&\hspace{4cm}
- \lambda^2 \det \begin{bmatrix} \zeta-V_2 & -\lambda & & \\ -\lambda & \ddots & \ddots & \\ & \ddots & \ddots & -\lambda \\ & & -\lambda & \zeta-V_{p-1} \end{bmatrix} - \lambda^p (z + 1/z).\notag
\end{eqnarray}
 Note that (\ref{charpol})  can be derived by standard methods (row/column expansions, multi-linearity).
By further expanding the first determinant in (\ref{charpol}), one can see that $F(\zeta)$ can be written as
\begin{equation}
F(\zeta)=  f(\zeta)+g(\zeta) 
\end{equation}
where 
\begin{equation}\label{def:f-g}
f(\zeta):= (\zeta-V_1) \ldots (\zeta-V_p)\quad \text{and } \quad
 g(\zeta):= \lambda^2 h_p(\lambda,\zeta) - \lambda^p (z+1/z). 
\end{equation}
Here $h_p(\lambda,\zeta)$ is a polynomial in $\zeta$ and $\lambda$ of degree $p$.
 Thus there is a $C=C(\Gamma, \gamma,p)>0$ such that $|h_p(\lambda,\zeta)|\le C$ uniformly in 
\begin{equation}\label{eq:lambda-zeta-range}
(\lambda,\zeta)\in [0,1]\times \bigcup_{j=1}^p\{\zeta\in \C;\ |\zeta-V_j|\leq \gamma/2\}.
\end{equation}
An exact formula for $h_p(\lambda,\zeta)$ is provided in (\ref{eq:h-p}) below.
Moreover, a numerically accessible value for $C=C(\Gamma, \gamma,p)$ is given in Lemma \ref{lem:h-p}.

It will turn out to be crucial that all lower order terms in  (\ref{charpol}) are at least quadratic in $\lambda$. We will show \eqref{lem:eq-lambda-eign} for $\lambda \in \R$ such that
\begin{equation} \label{lambdarestrict}
\lambda^2 \le \lambda_0^2=\min \left\{ 1, \frac{(\gamma/2)^p}{2\hat{C}} \right\}\ \text{ where }\ \hat{C}:=C+2\rho_0+ \frac{1}{2\rho_0}>1.
\end{equation}
This determines the value of $\mu_0$ in our results. For $\lambda$ as in \eqref{lambdarestrict} we choose
\begin{equation}\label{def:gamma0}
\gamma_0 := \frac{\lambda^2\hat{C}}{(\gamma/2)^{p-1}}\leq ^\eqref{lambdarestrict} \gamma/4.
\end{equation}

We now fix $\ell \in \{1,\ldots,p\}$, and we focus on the neighborhood  of $V_\ell$.  
We consider the functions $f(\zeta)$ and $g(\zeta)$ defined in (\ref{def:f-g}) over  the circle $\Gamma_{\gamma_0}(V_\ell)=\{\zeta \in \C;\ |\zeta-V_\ell |=\gamma_0\}$. 

As $|\zeta-V_k| > \gamma/2$ for all $k\not= \ell$, we have $|f(\zeta)| > \left(\gamma/2\right)^{p-1} \gamma_0$ as well as 
\begin{equation}
|g(\zeta)| \le \lambda^2\left(|h_p(\lambda,\zeta)| +\left|z+1/z\right|\right) \le \lambda^2 \hat{C}.
\end{equation}
By the original choice of $\gamma_0$ in (\ref{def:gamma0}), this means that $|g(\zeta)| < |f(\zeta)|$ for all $\zeta\in\Gamma_{\gamma_0}(V_\ell)$. Thus, by Rouch\'e's Theorem, see e.g. \cite{Rouche}, $F(\zeta) = f(\zeta)+g(\zeta)$ has a unique root  $\zeta_\ell(\lambda,z)$ (an eigenvalue of $\tilde{A}_p(\lambda,z)$) with 
\begin{equation} \label{zetabound}
|\zeta_\ell(\lambda,z)-V_\ell| < \gamma_0 \le \frac{\gamma}{4}
\end{equation}
and $\zeta_\ell(\lambda,z)$ is analytic in $\frac{1}{2\rho_0} < |z| < 2\rho_0$. 

The argument above works for all $\ell =1,\ldots,p$, so that we get all eigenvalues of $\tilde{A}_p(\lambda,z)$,
\begin{equation}
\sigma(\tilde{A}_p(\lambda,z)) = \{ \zeta_\ell(\lambda,z);\ \ell=1,\ldots,p\},
\end{equation}
such that (\ref{zetabound}) is satisfied for all $\ell=1,\ldots, p$. This proves Lemma \ref{lem:lambda-eign}.
\end{proof}

In the following, we give an explicit formula for the function $h_p(\lambda,\zeta)$ in (\ref{def:f-g}), and we quantify a bound for its modulus.

First, observe that  Lemma \ref{lem:determinant} shows that the first term (determinant) in (\ref{charpol}) is given by the formula (set $a_j=\zeta-V_j$ and $b_j=-\lambda$ for all $j=1,\ldots,p$ in (\ref{eq:determinant})) 
 \begin{equation}
\det \begin{bmatrix}
 \zeta-V_1 & -\lambda & & \\ -\lambda & \ddots & \ddots & \\ & \ddots & \ddots & -\lambda \\ & & -\lambda & \zeta-V_{p} 
\end{bmatrix}=
  \prod_{n=1}^p (\zeta-V_n) 
+\lambda^2 H(\zeta, \lambda, p)
\end{equation}
where
\begin{equation}\label{def:H}
H(\zeta, \lambda, p):=
\sum_{k=1}^{\lfloor\frac p2\rfloor}(-1)^k \lambda^{2k-2}
 \sum_{\tiny\begin{array}{c}
 1\leq j_1<\ldots<j_k <p;\\
\Omega_{j_1}, \ldots,\Omega_{j_k} \text{ are }\\
\text{ mutually disjoint}
 \end{array}} 
\prod_{\tiny\begin{array}{c}
m=1,\ldots,p\\
\displaystyle m\notin \bigcup_{\ell}\Omega_{j_\ell}
\end{array}
} (\zeta-V_m).
\end{equation}
Here $\Omega_k=\{k,k+1\}\subset\{1,2,\ldots,p\}$.

This shows with (\ref{charpol}) that the function $h_p(\lambda,\zeta)$ in (\ref{def:f-g}) is given as 
\begin{equation}\label{eq:h-p}
h_p(\lambda,\zeta)=H( \lambda,\zeta, p)-\det \begin{bmatrix} \zeta-V_2 & -\lambda & & \\ -\lambda & \ddots & \ddots & \\ & \ddots & \ddots & -\lambda \\ & & -\lambda & \zeta-V_{p-1} \end{bmatrix}.
\end{equation}

\begin{lem}\label{lem:h-p}
The function $h_p(\lambda,\zeta)$ in (\ref{def:f-g}) is bounded as
 \begin{equation}\label{eq:h-p2}
 |h_p(\lambda,\zeta)|\leq C:=  2^p(\Gamma+\gamma/2)^p+(\Gamma+\gamma/2+2)^{p-2},
 \end{equation}
 for all $\lambda,\zeta$ as in (\ref{eq:lambda-zeta-range}).
\end{lem}
\begin{proof}
Observe that (recall that $\Gamma:=\max_{j\neq k}|V_j-V_k|
$)
\begin{equation}
|\zeta-V_m|\leq \Gamma+\gamma/2
\end{equation}
for all $\zeta$ in the domain (\ref{eq:lambda-zeta-range}) and $m=1,\ldots, p$. 
Hence, $H( \lambda,\zeta, p)$ in (\ref{def:H}) is bounded as (where $\lambda<1$)
\begin{eqnarray}\label{eq:h-p-bound1}
|H( \lambda,\zeta, p)|&\leq& (\Gamma+\gamma/2)^p \sum_{k=1}^{\lfloor\frac p2\rfloor}
 \sum_{\tiny\begin{array}{c}
 1\leq j_1<j_2<\ldots<j_k <p;\\
\Omega_{j_1}, \ldots,\Omega_{j_k} \text{ are mutually disjoint}
 \end{array}} 
1 \notag\\
&<& (\Gamma+\gamma/2)^p\sum_{k=0}^{p}{p \choose k}=2^p(\Gamma+\gamma/2)^p.
\end{eqnarray}
 
Moreover, it is immediately seen that the norm of the matrix in the second term in (\ref{eq:h-p}) is bounded by $\Gamma+\gamma/2+2\lambda$. Thus,
\begin{equation}\label{eq:h-p-bound2}
\det \begin{bmatrix} \zeta-V_2 & -\lambda & & \\ -\lambda & \ddots & \ddots & \\ & \ddots & \ddots & -\lambda \\ & & -\lambda & \zeta-V_{p-1} \end{bmatrix}\leq (\Gamma+\gamma/2+2\lambda)^{p-2}\leq (\Gamma+\gamma/2+2)^{p-2}.
\end{equation}
(\ref{eq:h-p-bound1}) and (\ref{eq:h-p-bound2}) show that 
\begin{equation}\label{def:C}
|h_p(\lambda,\zeta)| \leq   2^p(\Gamma+\gamma/2)^p+(\Gamma+\gamma/2+2)^{p-2}.
\end{equation}
\end{proof}

\section{The asymptotic velocity}\label{sec:a-v}
In this section we prove our second main result, Thm \ref{thm:main-2}. So we are interested in finding an upper bound for the asymptotic velocity of the Schr\"odinger operator $H_\mu$ given in (\ref{model}) with the $p$-periodic potential $V:\Z\rightarrow\mathbb{R}$ such that
\begin{equation}
\gamma=\min_{0<|j-k|<p}|V_j-V_k|>0.
\end{equation}
The asymptotic velocity is given by the formula
\begin{equation}
v_{Asy}^{(\mu)}=\lim_{t\rightarrow\infty}\frac1t\big\|Xe^{-itH_\mu}\delta_0\big\|
\end{equation}
where $X$ is the position operator  (\ref{def:X}) on $\mathcal{H}$ and $\delta_0=|\e_1,0\rangle$. 
Recall that the Floquet transform (\ref{eq:F}) reduces the study of $H_\mu$ to the study of the corresponding $p\times p$ Floquet matrix $A_p(\mu,e^{ix})$ given in (\ref{def:Floquet-p}). As was the case in Section \ref{sec:LR-p}, instead of working with large $\mu>1$, we will work with the small $\lambda=1/\mu<1$, i.e., we consider
\begin{equation}\label{eq:Floquet}
\tilde{A}_p(\lambda, e^{ix})=\frac{1}{\mu}A_p(\mu, e^{ix}) =\begin{bmatrix} V_1 & \lambda & & & \lambda e^{ix} \\ \lambda &  V_2 & \ddots & & \\ & \ddots & \ddots & \ddots & \\ & & \ddots & V_{p-1} & \lambda \\ \lambda e^{-ix} & & & \lambda & V_p \end{bmatrix},  x\in[0,2\pi)
\end{equation}
for $p>2$ with the eigenvalues $\zeta_\ell(\lambda,e^{ix})$, for $\ell=1,2,\ldots,p$.

We remark here that for $p=2$ we have the slightly structurally different matrix
\begin{equation}
\tilde{A}_2(\lambda, e^{ix})=\begin{bmatrix}V_1 & \lambda(1+e^{ix})\\\lambda(1+e^{-ix})&  V_2\end{bmatrix}.
\end{equation}

Moreover, we will work in the setting where $\lambda$ is sufficiently small to make the gap between any two eigenvalues uniformly bounded in $p$. This is in particular satisfied for all $\lambda\leq \lambda_0$ where $\lambda_0$ is the same as the one in Lemma \ref{lem:lambda-eign}. For definiteness, $\lambda_0$ is given by the formula (\ref{lambdarestrict}). Note that this corresponds to consider working with $\mu\geq \mu_0$ where $\mu_0=1/\lambda_0$. Thus, Lemma \ref{lem:lambda-eign} implies that if $\lambda\leq \lambda_0$ then
\begin{equation}
\min_{j\neq k}|\zeta_j(\lambda,e^{ix})-\zeta_k(\lambda,e^{ix})|\geq \gamma/4 \text{ for all } x\in[0,2\pi).
\end{equation}
 Our approach uses the crucial result \cite[Thm 1.6]{DLY} that reduces the problem of finding the asymptotic velocity of periodic operators to the problem of finding (bounding) the derivatives of the eigenvalues of the corresponding Floquet matrix, see also extensions of the result to general graphs in \cite[Thm 2.3]{Asch-Knauf-1998} and \cite[Thm 3.1]{Monval-Sabri-22}. We note here that the following theorem, in particular the formula on the RHS of(\ref{eq:limit}) below, is essentially contained  in the proof of  \cite[Thm 1.6]{DLY}.

 \begin{thm}\cite[Thm 1.6]{DLY} \label{thm:DLY}
For any $\psi\in D(X)$,
 \begin{equation}\label{eq:limit}
 \lim_{t\rightarrow\infty}  \frac1t \tau_t^{H_\mu}(X)\psi= \F^{-1}\left(\int_0^{2\pi} p\mu \sum_{\ell=1}^p \big(\partial_x\zeta_\ell(\lambda,e^{ix})\big) P_\ell(x)\ dx\right)\F\psi
 \end{equation}
 where $P_\ell(x)$ is the spectral projection corresponding to the eigenvalue $\zeta_\ell(\lambda,e^{ix})$ for $\ell=1,\ldots,p$.
 \end{thm}
 In Thm \ref{thm:DLY},  $\tau_t^{H_\mu}(X)$ is the Heisenberg evolution of the position operator $X$,
$
 \tau_t^{H_\mu}(X)=e^{itH_\mu}X e^{-itH_\mu}
$,
and  $\F:\ell^2(\Z;\C^p)\rightarrow \mathcal{L}^2([0,2\pi);\C^p)$ is the componentwise Fourier transform defined in (\ref{eq:F}).
  
Observe that $\delta_0\in D(X)$ where $X\delta_0=0$, and $\F\delta_0=(2\pi)^{-1/2}\e_1\in\mathcal{L}^2([0,2\pi),\C^p)$, and since $e^{itH_\mu}$ is unitary, we obtain from Thm \ref{thm:DLY} that

\begin{eqnarray}\label{eq:vp-eig-der}
v_{Asy}^{(\mu)}=\lim_{t\rightarrow\infty}\frac1t\left\| X e^{-itH_\mu}\delta_0\right\|&=&\lim_{t\rightarrow\infty}\frac1t\left\|\tau_t^{H_\mu}(X) \delta_0\right\| \notag\\
&=&\left\|\int_0^{2\pi} p\mu \sum_{\ell=1}^p \big(\partial_x\zeta_\ell(\lambda,e^{ix})\big) P_\ell(x)\e_1\ dx\otimes |0\rangle \right\| \notag\\
&\leq&   p\mu\max_{x,\ell}\left|\partial_x\zeta_\ell(\lambda,e^{ix})\right| \sum_{\ell=1}^p\int_0^{2\pi} \|P_\ell(x)\e_1\|\ dx
\notag\\
&\leq&  2\pi p^2\mu \max_{x,\ell}\left|\partial_x\zeta_\ell(\lambda,e^{ix})\right|.\end{eqnarray}

Thus, we need to find practical (useful) bounds for the derivatives of eigenvalues of the Floquet matrix. We can find explicit formulas for the eigenvalues of $A_p(\mu, e^{ix})$ when $p$ is small. In particular, 
in the case of a $2\times 2$ Floquet matrix, the eigenvalues of $\tilde{A}_2(\lambda,e^{ix})$ are given by the formulas
\begin{equation}
\zeta_\ell(\lambda,e^{ix})=\frac12(V_1+V_2)+\frac{(-1)^\ell}{2}\sqrt{(V_1-V_2)^2+8\lambda^2(1+\cos(x))} \quad\text{ for } \ell=1,2,
\end{equation}
 from which we find that if $|V_1-V_2|=\gamma>0$ then
\begin{equation}
|\partial_x\zeta_\ell(\lambda,e^{ix})|^2=\frac{4\lambda^4\sin^2 x}{\gamma^2+8\lambda^2(1+\cos x)}\leq \frac{4\lambda^4}{\gamma^2}.
\end{equation}
Use this bound in (\ref{eq:vp-eig-der}) to obtain (compare with Thm \ref{thm:p2} where $\gamma=2$.)
\begin{equation}
v_{Asy}^{(\mu)}\leq \frac{16\pi}{\gamma \mu}.
\end{equation}

For general $p\times p$ matrices we use the following theorem.
\begin{thm}\label{thm:der-Eig-2}
For each $\ell=1,\ldots,p$ where $p>2$, the derivative of the eigenvalue $\zeta_\ell(\lambda,e^{ix})$ of $\tilde{A}_p(\lambda,e^{ix})$ in (\ref{eq:Floquet}) is given by the formula
\begin{equation}
\partial_x\zeta_\ell(\lambda,e^{ix})=-2\lambda^p\sin x \prod_{\tiny \begin{array}{c}j=1\\ j\neq \ell\end{array}}^p\big(\zeta_\ell(\lambda,e^{ix})-\zeta_j(\lambda,e^{ix})\big)^{-1}.
\end{equation}
\end{thm}
The proof of Theorem \ref{thm:der-Eig-2} is presented after the following argument.

Recall here that Lemma \ref{lem:lambda-eign} gives that for all $\lambda\leq \lambda_0$ (sufficiently small hopping terms or equivalently, sufficiently large potential $\mu\geq \mu_0$), each eigenvalue $\zeta_\ell(\lambda,e^{ix})$ has distance no more than $\gamma/4$ to $V_\ell$ for  $\ell=1,2,\ldots,p$. Note here that $z=e^{ix}$ is in the $z$-analycity range (\ref{eq:range-lambda-z}). Hence,  we have the (rough) lower bound
\begin{equation}\label{eq:lambdaPrime-D2}
\prod_{\tiny \begin{array}{c}j=1\\ j\neq \ell\end{array}}^p|\zeta_\ell(\lambda,e^{ix})-\zeta_j(\lambda,e^{ix})|\geq \left(\frac{\gamma}{2}\right)^{p-1}.
\end{equation}
Use this bound in the formula in Thm \ref{thm:der-Eig-2} to obtain
\begin{equation}
|\partial_x\zeta_\ell(\lambda,e^{ix})|\leq 2 \lambda^p\left(\frac{2}{\gamma}\right)^{p-1}\ \text{ for all }\ \ell=1,\ldots,p.
\end{equation}
Substitute in  (\ref{eq:vp-eig-der}) to find the bound
\begin{equation}
v_{Asy}^{(\mu)} \leq   \frac{C_3}{\mu^{p-1}} \text{ where }\ C_3:=4\pi p^2\left(\frac{2}{\gamma}\right)^{p-1}.
\end{equation}
To finish the proof of Thm \ref{thm:main-2}, it only remains  to prove Thm \ref{thm:der-Eig-2}.

\begin{proof}[Proof of Thm \ref{thm:der-Eig-2}]
First, $\zeta=\zeta_\ell(\lambda,e^{ix})$ (where $\zeta_\ell(\lambda,e^{ix})$ is any eigenvalue of $\tilde{A}_p:=\tilde{A}_p(\lambda,e^{ix})$ in (\ref{eq:Floquet})) satisfies the characteristic equation $\det(\zeta\idty-\tilde{A}_p)=0$. Take the derivative with respect to $x$ and use the Jacobi formula, see e.g., \cite{Magnus-Matrix-Calc}: For any $n\times n$ matrix $S(x)$,
\begin{equation}\label{def:JacobiF}
\frac{d}{dx}\det(S(x))=\Tr\left(\text{adj}(S(x))\frac{dS(x)}{dx}\right)
\end{equation}
where $dS(x)/dx$ is the $n\times n$ matrix whose $ij$-th entry is the derivative of the $ij$-th entry of $S(x)$, i.e., $\left(dS(x)/dx\right)_{i,j}= d\left((S(x)_{i,j}\right))/dx$.
We obtain  the equation (in $\zeta$)
\begin{equation}\label{eq:der-det}
\Tr\left(\text{adj}(\zeta\idty-\tilde{A}_p)\frac{d}{dx}(\zeta\idty-\tilde{A}_p)\right)=0.
\end{equation}

This equation is satisfied for all $\zeta=\zeta_\ell:=\zeta_\ell(\lambda,e^{ix})$ for $\ell\in\{1,\ldots,p\}$. Observe that
\begin{equation}
\frac{d}{dx}(\zeta_\ell\idty-\tilde{A}_p)=
\partial_x\zeta_\ell \idty+i\lambda e^{ix}|\e_1\rangle\langle \e_p|-i\lambda e^{-ix}|\e_p\rangle\langle \e_1|.
\end{equation}
Multiply by $\text{adj}(\zeta_\ell\idty-\tilde{A}_p)$, take the trace, then solve (\ref{eq:der-det}) for $\partial_x\zeta_\ell$ to obtain
\begin{equation}\label{eq:lambdaprime}
\partial_x \zeta_\ell(\lambda,e^{ix})=\frac{-i\lambda\left(\langle \e_1,\text{adj}(\zeta_\ell\idty-\tilde{A}_p)\e_p\rangle e^{-ix}-\langle \e_p,\text{adj}(\zeta_\ell\idty-\tilde{A}_p)\e_1\rangle e^{ix}\right)}{\Tr\left(\text{adj}(\zeta_\ell\idty-\tilde{A}_p)\right)}.
\end{equation}
Then,
\begin{eqnarray*}
\langle \e_1, \text{adj}(\zeta_\ell\idty-\tilde{A}_p)\e_p\rangle
&=&
(-1)^{p+1}
\det\begin{bmatrix}
-\lambda   &              &         & & -\lambda e^{ix}\\
\zeta_\ell-V_2& -\lambda &         &  &                      \\
-\lambda  & \ddots     &\ddots        & &           \\
                 & \ddots      & \ddots  &\ddots        &   \\
                 &&&&&\\
               &   & -\lambda  & \zeta_\ell-V_{p-1} & -\lambda   
\end{bmatrix}\\
&=&
\lambda^{p-1}+\lambda e^{ix}
\det\begin{bmatrix}
\zeta_\ell-V_2 & -\lambda &         &            \\
-\lambda  & \ddots     &\ddots        &            \\
                 & \ddots      & \ddots          & -\lambda \\
                &  & -\lambda  & \zeta_\ell-V_{p-1}   
\end{bmatrix}.
\end{eqnarray*}
In the last step, we expanded the determinant along the first row. Use the fact that 
\begin{equation}
\langle \e_1,\text{adj}(\zeta_\ell\idty-\tilde{A}_p)\e_p\rangle=\overline{\langle \e_p,\text{adj}(\zeta_\ell\idty-\tilde{A}_p)\e_1\rangle}
\end{equation}
to see that 
\begin{equation}\label{eq:lambdaPrime-N}
\langle \e_1,\text{adj}(\zeta_\ell\idty-\tilde{A}_p)\e_p\rangle e^{-ix}-\langle \e_p,\text{adj}(\zeta_\ell\idty-\tilde{A}_p)\e_1\rangle e^{ix}=
-2i\lambda^{p-1}\sin x.
\end{equation}
Moreover, we use Jacobi's formula (\ref{def:JacobiF}) with $S(t)=t\idty -\tilde{A}_p$ to obtain the identity
\begin{equation}
\frac{d}{dt}\det\left(t\idty-\tilde{A}_p\right)=\Tr\left(\text{adj}(t\idty-\tilde{A}_p) \right).
\end{equation}
Hence,
\begin{eqnarray}\label{eq:lambdaPrime-D}
\Tr\left(\text{adj}(\zeta_\ell \idty-\tilde{A}_p)\right)&=&\frac{d}{dt}\Big|_{t=\zeta_\ell}\det(t\idty-\tilde{A}_p) \nonumber\\
&=&\frac{d}{dt}\Big|_{t=\zeta_\ell}\prod_{j=1}^p (t-\zeta_j)
=
\prod_{\tiny \begin{array}{c}j=1\\ j\neq \ell\end{array}}^p(\zeta_\ell-\zeta_j).
\end{eqnarray}
Substitute (\ref{eq:lambdaPrime-N}) and (\ref{eq:lambdaPrime-D}) in (\ref{eq:lambdaprime}) to obtained the desired formula of $\partial_x \zeta_\ell(\lambda,e^{ix})$.
\end{proof}

\section*{Acknowledgments}

\begin{itemize}
   
    \item[]We would like to thank  Yulia Karpeshina and Mostafa Sabri for fruitful discussions.
    \item[] H. A. is supported in part by the UAE University under grant number G00004622.
\item[] G.S. would also like to thank Sir Isaac Newton (the younger) for emotional support.

\end{itemize}

\section*{Data Availability}
No data are associated with this article.

\appendix
\section{Lieb-Robinson velocity bound versus the asymptotic velocity}\label{sec:LR-ASY}
The following general Lemma shows that the Lieb-Robinson velocity is an upper bound for the asymptotic velocity of general (single-body) Hamiltonians. In particular, let $H$ be any self-adjoint operator on $\ell^2(\Z;\C^p)\cong \C^p\otimes\ell^2(\Z)$ for any $p\geq 1$. Let's recall the basic notations used along this paper. We denote the canonical basis of $\C^p$ and $\ell^2(\Z)$ by $\{\e_m,\ m=1,\ldots,p\}$ and $\{|j\rangle,\ j\in\Z\}$, respectively. Moreover, we use the notation $|\e_m,j\rangle:=|\e_m\rangle\otimes |j\rangle$ and $\delta_0:=|\e_1,0\rangle$. The position operation on $\C^p\otimes\ell^2(\Z)$ is defined as 
\begin{equation}
X=\idty_{\C^p}\otimes\sum_{j\in\Z}|j|\ |j\rangle\langle j|, \text{ i.e., }X|\e_m,j\rangle=|j|\ |\e_m,j\rangle \text{ for all }j\in\Z,\ m=1,\ldots,p.
\end{equation}
For $j,k\in\Z$, $e^{-itH }(j,k)\in\C^{p\times p}$ denotes the $(j,k)$-th block of $e^{-itH }$. We note here that in the case where $p=1$, we have $e^{-itH }(j,k)=\langle j| e^{-itH }|k\rangle$.
\begin{thm} \label{thm:LR-v}
If there are constants $C<\infty$, $\eta>0$, and $v>0$ such that
 \begin{equation}\label{eq:lem:LR}
 \|e^{-itH }(j,k)\|\leq C e^{-\eta(|j-k|-v|t|)},
 \end{equation}
 for all $j,k\in\Z$ and $t\in \mathbb{R}$.
 Then
 \begin{equation}\label{eq:lem:v}
 \lim_{t\rightarrow\infty}\frac1t\left\|X e^{-itH }\delta_0\right\|\leq v.
 \end{equation}
 \end{thm}
 \begin{proof}
 (\ref{eq:lem:LR}) shows dynamical localization for small times $t$. In particular, given any $\varepsilon\in(0,1)$,  
 \begin{equation}\label{eq:lem:small-t}
 \text{if } v|t|\leq (1-\varepsilon)|j|\text{ then }\|e^{-itH }(j,0)\|\leq C e^{-\eta\varepsilon|j|}.
 \end{equation}
 We proceed as follows
 \begin{eqnarray}
 \left\|X e^{-itH }\delta_0\right\|^2
 &=&
 \left\langle e^{-itH }\delta_0,\idty_{\C^p}\otimes\sum_{j\in\Z}j^2|j\rangle\langle j| e^{-itH }\delta_0\right\rangle \notag\\
 &=& \sum_{j\in\Z}\sum_{m=1}^p j^2\left|\langle \e_m,j|e^{-itH }|\delta_0\rangle\right|^2 \notag\\
 &=& \sum_{|j|<\frac{v|t|}{1-\varepsilon}}\sum_{m=1}^p j^2 \left|\langle \e_m,j|e^{-itH }|\delta_0\rangle\right|^2 +\sum_{|j|\geq \frac{v|t|}{1-\varepsilon}}\sum_{m=1}^p j^2 \left|\langle \e_m,j|e^{-itH }|\delta_0\rangle\right|^2.
 \end{eqnarray}
 We then replace $j$ in the first sum by its maximal value $v|t|/(1-\varepsilon)$, and in the second sum we use the dynamical localization (\ref{eq:lem:small-t}) while observing that
 \begin{equation}
 \sum_{m=1}^p  \left|\langle \e_m,j|e^{-itH }|\delta_0\rangle\right|^2=\left\|e^{-itH }(j,0)\e_1\right\|^2\leq \left\|e^{-itH }(j,0)\right\|^2.
 \end{equation}
 We obtain
 \begin{equation}
 \left\|X e^{-itH }\delta_0\right\|^2 \leq \left(\frac{v|t|}{1-\varepsilon}\right)^2\sum_{j\in\Z}\sum_{m=1}^p \left|\langle \e_m,j|e^{-itH }|\delta_0\rangle\right|^2 +C\sum_{j\in\Z} j^2 e^{-\eta\varepsilon|j|}.
 \end{equation}
 Here we note that
 \begin{equation}
 \sum_{j\in\Z}\sum_{m=1}^p \left|\langle \e_m,j|e^{-itH }|\delta_0\rangle\right|^2 =\|e^{-itH }\delta_0\|^2=1,
 \end{equation}
 and the value of the convergent series $\sum_{j} j^2 e^{-\eta\varepsilon|j|}$ depends only on $\epsilon$ and $\eta$, but not on $t$. Hence, 
 \begin{equation}
 C\sum_{j\in\Z} j^2 e^{-\eta\varepsilon|j|}=\tilde{C}(\varepsilon,\eta)<\infty.
 \end{equation}
 Thus, we have
 \begin{equation}
  \left\|X e^{-itH }\delta_0\right\|\leq \frac{v|t|}{1-\varepsilon}+ \left(\tilde{C}(\varepsilon,\eta)\right)^{1/2}.
 \end{equation}
 Here we used the fact that $\sqrt{a+b}\leq \sqrt{a}+\sqrt{b}$ for any positive numbers $a$ and $b$.
 Divide by $t$, take the limit as $t\rightarrow \infty$ then take the limit as $\varepsilon\rightarrow 0$ to get the desired bound.
  \end{proof}

 \section{The determinant of a Jacobi matrix}\label{sec:Jacobi}
 
Consider the $p\times p$ Jacobi matrix
 \begin{equation}\label{def:Jac}
 J_p(a,b):=\begin{bmatrix}
 a_1 & b_1 &  & &  \\
 \bar b_1    & a_{2} & b_2 & &  \\
       & \ddots & \ddots & \ddots & \\
       &&\bar b_{p-2}&a_{p-1}& b_{p-1}\\
       &&&\bar b_{p-1}& a_p
 \end{bmatrix}_{p\times p}.
 \end{equation}
 where $\{a_j\}_j$ and $\{b_j\}_j$ are arbitrary sequences of complex numbers.

 \begin{lem}\label{lem:determinant}
Let $\Lambda_p:=\{1,2,\ldots, p\}$, and for $j=1,2,\ldots,p-1$ let $\Omega_j\subset\Lambda_p$ be the sets of consecutive pairs, $\Omega_j :=\{j,j+1\}$. The determinant of $J_p(a,b)$ in (\ref{def:Jac}) is given by the following formula
 \begin{equation}\label{eq:determinant}
 \det(J_p(a,b))=\prod_{n=1}^p a_n +\sum_{k=1}^{\lfloor\frac p2\rfloor} 
 \sum_{\tiny\begin{array}{c}
 j_1,j_2\ldots,j_k \in\Lambda_{p-1};\\
\Omega_{j_1}, \ldots,\Omega_{j_k} \text{ are}\\
\text{ mutually disjoint}
 \end{array}} 
\prod_{\ell=1}^k (-|b_{j_\ell}|^2)
\prod_{\tiny\begin{array}{c}
m=1\\
\displaystyle m\notin \bigcup_{\ell=1,\ldots,k}\Omega_{j_\ell}
\end{array}
}^p a_m.
 \end{equation}
 \emph{Remark:} In the second term, when $k=1$, there will be only one $\Omega_j$ (where $j=1,\ldots,p-1$) for which the term ``mutually disjoint'' should be ignored.
 \end{lem}
While the explicit formula for $\det(J_p(a,b))$ in (\ref{eq:determinant}) looks somehow complicated, it is in fact  direct to apply when understood graphically. One needs to keep in mind a set of numbers (or vertices) labeled as $1,2,\ldots, p$, this is $\Lambda_p$. The first sum in (\ref{eq:determinant}) runs over $k$ and refers to the number of pairwise disjoint sets (of consecutive pairs) $\Omega_j=\{j,j+1\}\subset\Lambda_{p}$.  The second sum is to add over all possible ways to arrange the $k$ sets $\Omega_j$'s in $\Lambda_p$. The product of the $a$'s and $b$'s in each term is understood from each configuration of the $\Omega_j$'s. In particular, every $\Omega_j$ corresponds to the number $-|b_j|^2$ (with the same index $j$) in the product with all $a_m$'s where $m$ does not belong to any of $\Omega_j$'s, see Example \ref{ex:p-3-6} below. Thus, the sum over $k$ in $\det(J_p(a,b))$ is fully understood from the graphical illustration:
\begin{equation}
\begin{array}{cccccc}
\smash{\underbrace{\begin{matrix}
1 & 2&\ldots&p\end{matrix}}_{\tiny
\begin{matrix} \text{$k$  sets  $\Omega_j$'s} \\ j\in\Lambda_{p-1},\ k=1,\ldots, \lfloor \frac{p}{2}\rfloor\end{matrix}}}   \\
 \end{array}\\[1.5cm]
 \end{equation}

We remark here that (\ref{eq:determinant}) can be written as a one sum over $k$ where the trivial case $k=0$ corresponds to the product of all $a_j$'s.
\begin{equation}\label{eq:determinant:k0}
\det(J_p(a,b))=\sum_{k=0}^{\lfloor\frac p2\rfloor} 
 \sum_{\tiny\begin{array}{c}
 j_1,j_2\ldots,j_k \in\Lambda_{p-1};\\
\Omega_{j_1}, \ldots,\Omega_{j_k} \text{ are}\\
\text{ mutually disjoint}
 \end{array}} 
\prod_{\ell=1}^k (-|b_{j_\ell}|^2)
\prod_{\tiny\begin{array}{c}
m=1\\
\displaystyle m\notin \bigcup_{\ell=1,\ldots,k}\Omega_{j_\ell}
\end{array}
}^p a_m.
\end{equation}

\begin{ex}\label{ex:p-3-6}
We find $\det(J_p(a,b))$ for the two cases $p=3$ and $p=6$ using (\ref{eq:determinant}) or (\ref{eq:determinant:k0}).
 \begin{itemize}
 \item[(a)] \label{ex:p=3} When  $p=3$: The sum over $k$ in (\ref{eq:determinant}) produces just one value $k=1$, and here we have to consider $\Omega_1=\{1,2\}$ and $\Omega_2=\{2,3\}$ in the second sum, this gives directly that
\begin{equation}
\det(J_3(a,b))=a_1a_2a_3-|b_1|^2 a_3-|b_2|^2 a_1.
\end{equation}
\item[(b)] When $p=6$: The sum over $k$ is understood from Table \ref{table}.
$det(J_6(a,b))$ is the sum of the terms in the ``Products of $a$'s and $b$'s'' column.
\end{itemize}
\begin{center}
\begin{table}[h]
\begin{tabular}{ |c | c| c| c|}
\hline
$k$ & $\Omega$'s & Graph & Products of $a$'s and $b$'s\\
\hline
$k=0$ & No $\Omega_j$'s &{\tiny$\begin{array}{cccccc}1 & 2 &3 & 4&5 &6\\ \end{array}$}& $a_1 a_2 a_3 a_4 a_5 a_6$\\[0.3cm]
\hline
$k=1$ & $\Omega_1$ & {\tiny$\begin{array}{|cc|cccc}\cline{1-2}\textbf{1} & 2 &3 & 4&5 &6\\ \cline{1-2}\end{array}$}& $-|b_1|^2 a_3 a_4 a_5 a_6$\\ 
           & $\Omega_2$ & {\tiny$\begin{array}{c|cc|ccc}\cline{2-3}1 & \textbf{2} &3 & 4&5 &6\\ \cline{2-3}\end{array}$} &$-|b_2|^2 a_1 a_4 a_5 a_6$\\
           & $\Omega_3$ & {\tiny$\begin{array}{cc|cc|cc}\cline{3-4}1 & 2 &\textbf{3} & 4&5 &6\\ \cline{3-4}\end{array}$} &$-|b_3|^2 a_1 a_2 a_5 a_6$\\
           & $\Omega_4$ & {\tiny$\begin{array}{ccc|cc|c}\cline{4-5}1 & 2 &3 & \textbf{4}&5 &6\\ \cline{4-5}\end{array}$} &$-|b_4|^2 a_1 a_2 a_3 a_6$\\
           & $\Omega_5$ & {\tiny$\begin{array}{cccc|cc|}\cline{5-6}1 & 2 &3 & 4&\textbf{5} &6\\ \cline{5-6}\end{array}$}  &$-|b_5|^2 a_1 a_2 a_3 a_4$\\[0.3cm]
\hline      
$k=2$    & $\Omega_1, \Omega_3$ & {\tiny$\begin{array}{|cc| |cc|cc}\cline{1-2} \cline{3-4}\textbf{1} & 2 &\textbf{3} & 4&5 &6\\ \cline{1-2} \cline{3-4}\end{array}$}& $|b_1|^2 |b_3|^2 a_5 a_6$\\ 
    & $\Omega_1, \Omega_4$ &  {\tiny$\begin{array}{|cc| c|cc|c}\cline{1-2} \cline{4-5}\textbf{1} & 2 &3 & \textbf{4}&5 &6\\ \cline{1-2} \cline{4-5}\end{array}$} &$|b_1|^2 |b_4|^2 a_3 a_6$\\ 
    & $\Omega_1, \Omega_5$ &   {\tiny$\begin{array}{|cc| cc|cc|}\cline{1-2} \cline{5-6}\textbf{1} & 2 &3 & 4&\textbf{5} &6\\ \cline{1-2} \cline{5-6}\end{array}$}  &$|b_1|^2 |b_5|^2 a_3 a_4$\\ 
    & $\Omega_2, \Omega_4$ &   {\tiny$\begin{array}{c|cc||cc|c}\cline{2-3} \cline{4-5}1 & \textbf{2} &3 & \textbf{4}&5 &6\\ \cline{2-3} \cline{4-5}\end{array}$}  &$|b_2|^2 |b_4|^2 a_1 a_6$\\ 
     & $\Omega_2, \Omega_5$ &   {\tiny$\begin{array}{c|cc|c|cc|}\cline{2-3} \cline{5-6}1 & \textbf{2} &3 & 4&\textbf{5} &6\\ \cline{2-3} \cline{5-6}\end{array}$}  &$|b_2|^2 |b_3|^2 a_1 a_4$\\ 
     & $\Omega_3, \Omega_4$ &  {\tiny$\begin{array}{cc|cc||cc|}\cline{3-4} \cline{5-6}1 & 2 &\textbf{3} & 4&\textbf{5} &6\\ \cline{3-4} \cline{5-6}\end{array}$}  &$|b_3|^2 |b_5|^2 a_1 a_2$\\[0.3cm]    
    \hline
$k=3$ & $\Omega_1, \Omega_3,\Omega_5$&  {\tiny$\begin{array}{|cc||cc||cc|}\cline{1-2}\cline{3-4} \cline{5-6}\textbf{1} & 2 &\textbf{3} & 4&\textbf{5} &6\\ \cline{1-2}\cline{3-4} \cline{5-6}\end{array}$} & $-|b_1|^2 |b_3|^2  |b_5|^2$\\[0.3cm]
\hline
\end{tabular}
\caption{The determinant of a $6\times 6$ Jacobi matrix.}
\label{table}
\end{table}
\end{center}

\end{ex}

\begin{proof}[Proof of Lemma \ref{lem:determinant}]
 We proof formula (\ref{eq:determinant}) by induction on $p$.
 
 The base step $p=2$, results exactly one $\Omega=\{1,2\}$, and hence formula (\ref{eq:determinant}) gives the formula
 \begin{equation}
 \det(J_2(a,b))=a_1 a_2-|b_1|^2
 \end{equation}
as it should be.

Next, we assume that (\ref{eq:determinant}) is correct for all $m\leq p$ and we need to show that is valid for $m=p+1$, i.e., we want to prove that
\begin{equation}\label{eq:determinant-p+1}
 \det(J_{p+1}(a,b))=\prod_{n=1}^{p+1} a_n +\sum_{k=1}^{\lfloor\frac{p+1}2\rfloor}  
 \sum_{\tiny\begin{array}{c}
 j_1,j_2\ldots,j_k \in\Lambda_{p};\\
\Omega_{j_1}, \ldots,\Omega_{j_k} \text{ are}\\
\text{ mutually disjoint}
 \end{array}} 
\prod_{\ell=1}^k (-|b_{j_\ell}|^2)
\prod_{\tiny\begin{array}{c}
m=1\\
\displaystyle m\notin \bigcup_{\ell}\Omega_{j_\ell}
\end{array}
}^{p+1} a_m.
 \end{equation}

By expanding $ \det(J_{p+1}(a,b))$ along the last row, it is direct to see that
\begin{equation}
 \det(J_{p+1}(a,b))=a_{p+1}\det(J_p(a,b))-|b_p|^2\det(J_{p-1}(a,b)).
\end{equation}
Use (\ref{eq:determinant}) for $\det(J_p(a,b))$ and $\det(J_{p-1}(a,b))$ to obtain
\begin{eqnarray}\label{pf:determinant:p+1}
 \det(J_{p+1}(a,b))&=& \prod_{n=1}^{p+1} a_n +\sum_{k=1}^{\lfloor\frac{p}2\rfloor}
 \sum_{\tiny\begin{array}{c}
 j_1,j_2\ldots,j_k \in\Lambda_{p-1};\\
\Omega_{j_1}, \ldots,\Omega_{j_k} \text{ are}\\
\text{ mutually disjoint}
 \end{array}} 
\prod_{\ell=1}^k (-|b_{j_\ell}|^2)
\prod_{\tiny\begin{array}{c}
m=1\\
\displaystyle m\notin \bigcup_{\ell}\Omega_{j_\ell}
\end{array}
}^{p+1} a_m+\notag\\
&& -|b_p|^2\prod_{n=1}^{p-1} a_n -|b_p|^2\sum_{k=1}^{\lfloor\frac{p-1}2\rfloor} 
 \sum_{\tiny\begin{array}{c}
 j_1,j_2\ldots,j_k \in\Lambda_{p-2};\\
\Omega_{j_1}, \ldots,\Omega_{j_k}  \text{ are}\\
\text{ mutually disjoint}
 \end{array}} 
\prod_{\ell=1}^k (-|b_{j_\ell}|^2)
\prod_{\tiny\begin{array}{c}
m=1\\
\displaystyle m\notin \bigcup_{\ell}\Omega_{j_\ell}
\end{array}
}^{p-1} a_m.
\end{eqnarray}
The following argument explains why (\ref{pf:determinant:p+1}) is 
 equal to the desired formula (\ref{eq:determinant-p+1}):  The sum over $k$ in (\ref{eq:determinant-p+1}) can be understood from the graphical illustration:
\begin{equation}
\begin{array}{cccccc}
\smash{\underbrace{\begin{matrix}
1 & 2&\ldots&p &(p+1)\end{matrix}}_{\tiny
\begin{matrix} \text{$k$  sets  $\Omega_j$'s} \\ j\in\Lambda_p,\ k=1,\ldots, \lfloor \frac{p+1}{2}\rfloor\end{matrix}}}   \\
 \end{array}\\[1.5cm]
 \end{equation}
Let's recall that is understood as follows: The sum is taken over $k$ pairwise disjoint sets $\Omega_j=\{j,j+1\}$ where $j\in\Lambda_p$ and $k$ takes the values $1,\ldots, \lfloor \frac{p+1}{2}\rfloor$. Observe that this sum can be decomposed into the following three terms.
 \begin{equation}
\begin{array}{cccccc}
\smash{\underbrace{\begin{matrix}
1 & 2&\ldots&p\end{matrix}}_{\tiny
\begin{matrix} \text{$k$  sets  $\Omega_j$'s}\\ j\in\Lambda_{p-1}\\
k=1,\ldots, \lfloor \frac{p}{2}\rfloor\end{matrix}}}  
&(p+1)\\
 \end{array}+
\begin{array}{cccccc}
\smash{\underbrace{\begin{matrix}
1 & 2&\ldots&(p-1)\end{matrix}}_{\text{no $\Omega_j$'s}}}  &\smash{\underbrace{\begin{matrix} p & (p+1)\end{matrix}}_{\Omega_p}}\\
 \end{array} +
 \begin{array}{cccccc}
\smash{\underbrace{\begin{matrix}
1 & 2&\ldots&(p-1)\end{matrix}}_{\tiny\begin{matrix}\text{$k$ sets $\Omega_j$'s}\\
j\in\Lambda_{p-2}\\ 
k=1,\ldots, \lfloor\frac{p-1}{2}\rfloor\end{matrix}
}}  &\smash{\underbrace{\begin{matrix} p & (p+1)\end{matrix}}_{\Omega_p}}\\
 \end{array}\\[2cm]
\end{equation}
which match the second, third, and forth terms, respectively, in (\ref{pf:determinant:p+1}).
 \end{proof}



\end{document}